\documentclass[copyright,creativecommons]{eptcs}
\providecommand{\event}{QPL 2018} 

\usepackage{mathtools}
\usepackage{float}
\usepackage{hyperref}
\usepackage{fixltx2e}
\usepackage[framemethod=tikz]{mdframed}
\usepackage{multicol}
\usepackage{amsfonts}
\usepackage{amsmath}
\usepackage{amsthm}
\usepackage{enumitem}
\usepackage{tikz}
\usepackage{float}
\usepackage[all]{xy}
\usepackage{color,soul}
\usepackage{afterpage}

\newdir{ >}{{}*!/-8pt/@{>}}

\theoremstyle{definition}
\newtheorem{theorem}{Theorem}[section]
\newtheorem{corollary}[theorem]{Corollary}
\newtheorem{definition}[theorem]{Definition}
\newtheorem{lemma}[theorem]{Lemma}
\newtheorem{proposition}[theorem]{Proposition}

\renewcommand{\bar}{\overline}

\newcommand{\Total}{\mathsf{Total}}
\newcommand{\Par}{\mathsf{Par}}
\newcommand{\ParIso}{\mathsf{ParIso}}
\newcommand{\Set}{\mathsf{Set}}
\newcommand{\TOF}{\mathsf{TOF}}
\newcommand{\CNOT}{\mathsf{CNOT}}
\newcommand{\cnot}{\mathsf{cnot}}
\newcommand{\tof}{\mathsf{tof}}
\newcommand{\notgate}{\mathsf{not}}
\newcommand{\op}{\mathsf{op}}
\newcommand{\cnv}{\circ}
\newcommand{\ox}{\otimes}
\newcommand{\X}{\mathbb{X}}
\newcommand{\Y}{\mathbb{Y}}
\newcommand{\N}{\mathbb{N}}

\newcommand{\C}{\mathbb{C}}
\newcommand{\Z}{\mathbb{Z}}
\newcommand{\Pinj}{\mathsf{Pinj}}
\newcommand{\FPinj}{\mathsf{FPinj}}

\newcommand{\FHilb}{\mathsf{FHilb}}
\newcommand{\Mat}{\mathsf{Mat}}
\newcommand{\Not}{\mathsf{not}}

\newcommand{\oa}{\oplus}
\newcommand{\<}{\langle}
\newcommand{\la}{\langle}
\newcommand{\ra}{\rangle}
\renewcommand{\>}{\rangle}

\newcommand{\onein}{|1\>}
\newcommand{\oneout}{\<1|}

\newcommand{\zeroin}{|0\>}
\newcommand{\zeroout}{\<0|}

\newcommand\scalemath[2]{\scalebox{#1}{\mbox{\ensuremath{\displaystyle #2}}}}
\newcommand{\myeq}[1]{\stackrel{{\normalfont\scalebox{.75}{#1}}}{=}}

\tikzstyle{strings}=[baseline={([yshift=-.5ex]current bounding box.center)}]

\tikzset{every picture/.append style={scale=.5}, transform shape, strings}

\tikzset{%
symbol/.style={%
draw=none,
every to/.append style={%
edge node={node [sloped, allow upside down, auto=false]{$#1$}}}
}
}

\usetikzlibrary{shapes.geometric}
\usetikzlibrary{patterns}
\usetikzlibrary{fit}
\usetikzlibrary{positioning}
\usetikzlibrary{calc}
\usetikzlibrary{arrows}
\usetikzlibrary{decorations.markings}
\usetikzlibrary{decorations.pathreplacing}
\usetikzlibrary{shapes}

\pgfdeclarelayer{nodelayer}
\pgfdeclarelayer{edgelayer}
\pgfsetlayers{edgelayer,nodelayer,main}

\tikzset{simple/.style={}}
\tikzset{nothing/.style={outer sep=-3.4pt}}
\tikzset{map/.style={draw,fill=white, rectangle}}

\tikzset{dot/.style={thick, fill=black, circle, scale=1, inner sep = .05cm}}

\tikzset{oplus/.style={draw, scale=0.9,minimum height=.1cm,circle,append after command={
[shorten >=\pgflinewidth, shorten <=\pgflinewidth,]
(\tikzlastnode.north) edge (\tikzlastnode.south)
(\tikzlastnode.east) edge (\tikzlastnode.west)
}
}
}

\tikzset{fanin/.style={
draw,
shape border rotate=30,
regular polygon,
regular polygon sides=3,
fill=white,
inner sep = .1cm
}
}

\tikzset{fanout/.style={
draw,
shape border rotate=-30,
regular polygon,
regular polygon sides=3,
fill=white,
inner sep = .1cm
}
}

\tikzset{onein/.style={
draw,
shape border rotate=30,
regular polygon,
regular polygon sides=3,
fill=black,
inner sep = .04cm,
scale=1.2
}
}

\tikzset{oneout/.style={
draw,
shape border rotate=-30,
regular polygon,
regular polygon sides=3,
fill=black,
inner sep = .04cm,
scale=1.2
}
}

\tikzset{zeroin/.style={
draw,
shape border rotate=30,
regular polygon,
regular polygon sides=3,
fill=white,
inner sep = .04cm,
scale=1.2
}
}

\tikzset{zeroout/.style={
draw,
shape border rotate=-30,
regular polygon,
regular polygon sides=3,
fill=white,
inner sep = .04cm,
scale=1.2
}
}

\tikzset{wires/.style={}}

\tikzset{box/.style={inner sep=0pt, thick, draw=black, text height=1.5ex, text depth=.25ex, text centered, minimum height=3em, anchor=center}}

\makeatletter
\newcommand{\oneraggedpage}{\let\mytextbottom\@textbottom
	\let\mytexttop\@texttop
	\raggedbottom
	\afterpage{%
		\global\let\@textbottom\mytextbottom
		\global\let\@texttop\mytexttop}}

\title{The Category $\TOF$}
\author{
	J.R.B. Cockett  \qquad\qquad Cole Comfort
	\email{\{robin,crcomfor\}@ucalgary.ca}
	\institute{Department of Computer Science\\University of Calgary\\Alberta, Canada}
}

\def\titlerunning{The Category $\TOF$}
\def\authorrunning{J.R.B. Cockett \& Cole Comfort}

\begin{document}
\maketitle

\begin{abstract}
We provide a complete set of identities for the symmetric monoidal category, $\TOF$, generated by the Toffoli gate and computational ancillary bits.  We do so by demonstrating that the functor which evaluates circuits on total points, is an equivalence into the full subcategory of sets and partial isomorphisms with objects finite powers of the two element set.  The structure  of the proof builds -- and follows the proof of Cockett et al. -- which provided a full set of identities for the $\cnot$ gate with computational ancillary bits. Thus, first it is shown that $\TOF$ is  a discrete inverse category in which all of the identities for the $\cnot$ gate hold; and then a normal form for the restriction idempotents is constructed which corresponds precisely to subobjects of the total points of $\TOF$. This is then used to show that $\TOF$ is equivalent to $\FPinj_2$, the full subcategory of sets and partial isomorphisms in which objects have cardinality $2^n$ for some $n \in \mathbb{N}$.
\end{abstract}

\section{Introduction}

The Toffoli gate is a cornerstone for reversible computing: alongside the Fredkin gate, it was the first gate which was proven to be universal for classical reversible computing \cite{fredkin2002conservative}.  That is, if the values of certain wires are fixed and unchanged by computation (called ``auxiliary bits'') and the value of others ignored after computation (called ``garbage bits''), then every reversible boolean function can be simulated using either Fredkin or Toffoli gates.  The universality of the Fredkin gate relies fundamentally on the use of garbage bits.  The Toffoli gate, however, is universal for classical reversible circuits even when \emph{only} auxiliary bits are allowed.  This is notable, as even the Fredkin and $\cnot$ gates combined are not universal in this sense \cite[Thm. 3]{aaronson}.  Moreover, the Toffoli gate is also used frequently in quantum computation (especially in  quantum error correction \cite{Shor, Gottesman}); and has even been physically realized \cite{realize,realize2}.

Notably, \cite{Iwama} provided an infinite, complete set of identities for functions of the form $|x_1,\cdots, x_n, y\>\mapsto |x_1,\cdots, x_n, y+f(x_1,\cdots, x_n)\>$ generated by Toffoli gates with finitely  many control wires, along with finitely many qubits in the state $|0\>$.

Auxiliary bits are a rather peculiar notion and are unnatural in the context of symmetric monoidal categories of circuits
\footnote{An auxiliary bit can be simulated by fixing the input and output of a wire using an input and output ancillary bit.  However, usually, there is the added assumption that fixing bits in this manner will not cause the resulting function to degenerate or become partially defined.  Thus, auxiliary bits, as commonly formulated, do not provide a compositional notion.}
; instead, we shall take a more general and richer approach by simply allowing the inputs and outputs of wires to be fixed by components which we call ``ancillary bits''.  Ancillary bits, in this sense, are modelled by qubit initialization and termination in \cite{quipper}, where the state of a circuit can degenerate to an undefined state -- in finite dimensional Hilbert spaces, $\FHilb$, this corresponds to when maps compose to a zero-matrix.


In \cite{CNOT}, Cockett et al. provided a complete set of identities for the fragment of quantum computing generated by the controlled-not gate with ancillary bits.  The paper was inspired by Lafont's work  \cite{Lafont}, but used the notion of a discrete inverse category developed in \cite{Giles}.  Inverse categories have now been used  extensively to model the semantics of reversible computation \cite{inverse1,inverse2}.  The current paper builds on \cite{CNOT} and provides a finite, complete set of identities for the fragment of quantum computing generated by the Toffoli gate with ancillary bits \footnote{Although the larger fragment of quantum computing generated by the Toffoli gate, Hadamard gate and computational ancillary bits has recently been classified \cite{ZX}.}.  We call the symmetric monoidal category generated by the Toffoli gate and ancillary bits with these identities, $\TOF$.

In addition to providing a complete set of identities for these circuits, we also prove a concrete equivalence into the subcategory of sets and partial isomorphisms where the objects are finite powers of the 2 element set \footnote{This result implies, in particular, that $\TOF$ embeds  faithfully into the categories of matrices \cite{Barr}, $\Mat(\C)$, as $\FPinj_2$ embeds into the category of matrices by taking partial isomorphisms to their adjacency matrices.}.  The key step of this proof is to prove that the functor $\tilde H_0$, which takes an object to its (total) points, is faithful: this, in turn, relies on providing a normal form for the restriction idempotents of $\TOF$. 

\medskip
\noindent
{\bf Overview of the proof}

We first present the symmetric monoidal category called $\TOF$ which is generated by the Toffoli gate and ancillary bits along with 17 identities.  The paper culminates with a proof that $\TOF$ is isomorphic to the category, $\FPinj_2$, of partial isomorphisms between finite powers of the two element set.   The proof follows the form of the proof in \cite{CNOT} for the category $\CNOT$, build from the computational ancillary bits and the $\cnot$ gate.   We start by observing, Lemma \ref{lemma:functorial}, that all of the identities of the category $\CNOT$ hold in $\TOF$.  The first crucial step is to prove that some needed identities of Iwama \cite{Iwama} hold in this setting.  The next crucial next step is to prove that $\TOF$ is a discrete inverse category.   Discrete inverse  categories have inverse products and, for $\TOF$, these are essentially inherited from $\CNOT$. Furthermore, in $\TOF$, just as for $\CNOT$, partial inverses are given by horizontally flipping  circuits:  this gives a dagger functor $(\_)^\cnv:\TOF^\op\to \TOF$.   In Section \ref{section:thefunctor}, we construct a discrete inverse functor $\tilde H_0:\TOF\to\FPinj_2$ which ``evaluates''  maps on the points of $\TOF$.   In Section \ref{section:NormalForm}, we construct a normal form for the restriction idempotents of $\TOF$ using ``polyforms'': this is the crux of the paper.  In Section \ref{section:Equivalence}  we prove that the functor $\tilde H_0$ is full faithful, borrowing results from \cite{CNOT} --  wherein it is shown that $\tilde H_0$ being full and faithful  on restriction idempotents implies $\tilde H_0$ is full and faithful in general.


\section{Restriction and inverse categories}
\label{section:TOFINVERSE}


In this section, we recall the basic theory and terminology of restriction and inverse categories which will be used later.

\begin{definition}\cite[Def. 2.1.1]{Cockett}
A \textbf{restriction structure} on a category $\mathbb{X}$ is an assignment $\bar{f}: A\to A$ for each map $f:A\to B$ in $\mathbb{X}$ satisfying the following four axioms:

\begin{multicols}{4}
\begin{enumerate}[label={\bf [R.\arabic*]}, ref={\bf [R.\arabic*]}, wide = 0pt, leftmargin = 2em]
\item	\label{R.1}
		$\bar{f}f =f$

\item	\label{R.2}
		$\bar{f}\bar{g}=\bar{g}\bar{f}$

\item	\label{R.3}
		$\bar{\bar{g}f}=\bar{f}\bar{g}$

\item	\label{R.4}
		$f \bar{g}=\bar{fg}f$
\end{enumerate}
\end{multicols}

\end{definition}

A \textbf{restriction category} is a category equipped with a restriction structure. A \textbf{restriction functor} is a functor which preserves the given restriction structure.  An endomorphism $e: A \rightarrow A$ is called a \textbf{restriction idempotent} if $e=\bar{e}$.  In particular, each $\bar{f}$ has $\bar{\bar{f}} = \bar{f}$ and is an idempotent as $\bar{f}\bar{f} = \bar{\bar{f}f} = \bar{f}$.

In a restriction category, a \textbf{total map} is a map $f$ such that $\bar{f}=1$.  The total maps of a restriction category $\mathbb{X}$ form a subcategory $\Total(\mathbb{X})$ of $\mathbb{X}$. 


A basic example of a restriction category is the category of sets and  partial functions, $\Par$: the restriction of a partial function is the partial identity on the domain of definition.

\begin{definition}\cite[Sec. 2.3]{Cockett}.
A map $f$ is a \textbf{partial isomorphism} when there exists another map $g$, called the \textbf{partial inverse} of $f$, such that $\bar f = fg$ and $\bar g = gf$. 
A restriction category $\mathbb{X}$ is an \textbf{inverse category} when all its maps are partial isomorphisms. 
\end{definition}

Partial isomorphisms generalize the notion of isomorphisms to restriction categories; thus, the composition of partial isomorphisms is a partial isomorphism and partial inverses  are unique.  Furthermore, every restriction category $\mathbb{X}$ has a subcategory of partial isomorphisms $\ParIso(\mathbb{X})$ which is an inverse category.  Denote the category $\ParIso(\Set)$ by $\Pinj$.

There is an important alternate characterization of an inverse category:

\begin{theorem}\cite[Thm. 2.20]{Cockett}  \label{defn:inverseCategory}
	A category $\mathbb{X}$ is an inverse category if and only if there exists an functor ${(\_)^\cnv:\mathbb{X}^\op \to\mathbb{X}}$ which is the identity on objects, satisfying the following three axioms:
	
\begin{multicols}{3}
\begin{enumerate}[label={\bf [INV.\arabic*]}, ref={\bf [INV.\arabic*]}, wide = 0pt, leftmargin = 2em]
\item	\label{INV.1}
		$(f^\cnv)^\cnv = f$

\item	\label{INV.2}
		$ff^\cnv f = f$

\item	\label{INV.3}
		$ff^\cnv gg^\cnv = gg^\cnv ff^\cnv$

\end{enumerate}
\end{multicols}
\end{theorem}

Inverse categories have restriction structure given by $\bar c := cc^\cnv$.  It is not hard to show that every idempotent in an inverse category is necessarily a restriction idempotent.

Inverse categories can have a product-like structure:

\begin{definition}\cite[Def. 4.3.1]{Giles} 

Take a symmetric monoidal inverse category $\mathbb{X}$ with tensor ${\_\otimes\_:\mathbb{X}\times\mathbb{X}\to \mathbb{X}}$,  symmetry $c$ and associator $a$.  Suppose moreover that the tensor preserves $(\_)^\cnv$.  We say $\mathbb{X}$ has {\bf inverse products} if there exists a total natural diagonal transformation $\Delta$ which satisfies the following properties:

\begin{figure}[h]
\begin{multicols}{2}
\begin{enumerate}[label= {\bf [DINV.\arabic*]}, ref={\bf [DINV.\arabic*]},wide = 0pt, leftmargin = 2em]
\item	\label{DINV.1}
		$\Delta$ is cocommutative:
		\[  \xymatrix{ A \ar[dr]_{\Delta_A} \ar[r]^{\Delta_A} &A\otimes A\ar[d]^{c_{A,A}}\\ & A\otimes A} \]
\item	\label{DINV.2}
		$\Delta$ is coassociative:
		\[
		\xymatrix{
			A \ar[rr]^{\Delta_A} \ar[d]_{\Delta_A} && A\otimes A \ar[d]^{1_A\otimes \Delta_A}\\
			A\otimes A \ar[dr]_{\Delta_A\otimes 1_A} & & A\otimes(A\otimes A)\\
			& (A\otimes A)\otimes A \ar[ur]_{a_{A,A,A}} &
		}
		\]
\item	\label{DINV.3}
		$(\Delta,\Delta^\cnv)$ is a semi-Frobenius object:
		\[
		\xymatrix{
			A\otimes A \ar[rr]^{(\Delta_A\otimes 1_A)a_{A,A,A}} \ar[dr]^{\Delta_A^\cnv} \ar[dd]_{(1_A\otimes \Delta_A)a_{A,A,A}^\cnv} && A\otimes(A\otimes A) \ar[dd]^{1_A\otimes \Delta_A^\cnv}\\
			& A  \ar[dr]^{\Delta_A} &\\
			(A\otimes A)\otimes A \ar[rr]^{\Delta_A^\cnv\otimes1_A} && A\otimes A
		}
		\]
\item	\label{DINV.4}
		$\Delta$ is uniform-copying:
		\[
		\xymatrixrowsep{1.2cm}
		\xymatrix{
			A\otimes B \ar[r]^-{\Delta_A\otimes \Delta_B} \ar[dr]_{\Delta_{A\otimes B}}& (A\otimes A)\otimes (B\otimes B) \ar[d]^{
			\scalemath{.75}{%
			\begin{array}{l}
			a_{A,A,B\otimes B} (1_A\otimes a_{A,B,B}^\cnv)\\
			\ (1_A\otimes (c_{A,B} \otimes 1_B))\\
			\ ((1_A\otimes a_{B,A,B})a_{A,B,A\otimes B}^\cnv)
			\end{array}
			}}\\
			&(A\otimes B)\otimes (A\otimes B)
		}
		\]
\end{enumerate}
\end{multicols}
\end{figure}

\end{definition}


A {\bf discrete inverse category} is a category with inverse products.  Note that $\Delta$ is required to be total, so that $\Delta \Delta^\cnv = 1$: making
the semi-Frobenius structure separable (or special). 

\section{The category \texorpdfstring{$\CNOT$}{CNOT}}

In \cite{CNOT} the category $\CNOT$ was presented as a symmetric monoidal category with objects natural numbers generated by  the 1 ancillary bits $\onein\equiv
\begin{tikzpicture}
			\begin{pgfonlayer}{nodelayer}
			\node [style=onein] (0) at (0, -0) {};
			\node [style=nothing] (1) at (1, -0) {};
			\end{pgfonlayer}
			\begin{pgfonlayer}{edgelayer}
			\draw (0) to (1);
			\end{pgfonlayer}
			\end{tikzpicture}$,
and $\oneout\equiv
\begin{tikzpicture}
			\begin{pgfonlayer}{nodelayer}
			\node [style=nothing] (0) at (0, -0) {};
			\node [style=oneout] (1) at (1, -0) {};
			\end{pgfonlayer}
			\begin{pgfonlayer}{edgelayer}
			\draw (0) to (1);
			\end{pgfonlayer}
			\end{tikzpicture}$ \footnote{Note that the bra-ket notation clashes with diagrammatic composition because the inner product look like the outer product, and vice versa.  We therefore make the use of classical composition explicit when we use bra-ket notation with the symbol $\_\circ\_$, instead of concatenation, which we reserve for diagrammatic composition.},
and the controlled not gate $\cnot\equiv
\begin{tikzpicture}
	\begin{pgfonlayer}{nodelayer}
		\node [style=nothing] (0) at (0, -0) {};
		\node [style=nothing] (1) at (0, 0.5000001) {};
		\node [style=nothing] (2) at (1, 0.5000001) {};
		\node [style=nothing] (3) at (1, -0) {};
		\node [style=oplus] (4) at (0.5000001, -0) {};
		\node [style=dot] (5) at (0.5000001, 0.5000001) {};
	\end{pgfonlayer}
	\begin{pgfonlayer}{edgelayer}
		\draw (1) to (5);
		\draw (5) to (2);
		\draw (3) to (4);
		\draw (4) to (0);
		\draw (4) to (5);
	\end{pgfonlayer}
	\end{tikzpicture}$, where these gates satisfy the identities in Figure \ref{fig:CNOT}:

\begin{figure}[h]
	\noindent
	\scalebox{2.0}{%
		\vbox{%
			\begin{mdframed}[backgroundcolor=black!2,roundcorner=10.0pt, innerleftmargin=+.75cm]
				\begin{multicols}{2}
					\begin{enumerate}[label={\bf [CNOT.\arabic*]}, ref={\bf [CNOT.\arabic*]}, wide = 0pt, leftmargin = 2em]
						\item
						\label{CNOT.1}
						{\hfil
							$
							\begin{tikzpicture}
							\begin{pgfonlayer}{nodelayer}
							\node [style=nothing] (0) at (0, 0) {};
							\node [style=nothing] (1) at (0, .5) {};
							\node [style=oplus] (2) at (.5, 0) {};
							\node [style=dot] (3) at (.5, .5) {};
							\node [style=dot] (4) at (1, 0) {};
							\node [style=oplus] (5) at (1, .5) {};
							\node [style=oplus] (6) at (1.5, 0) {};
							\node [style=dot] (7) at (1.5, .5) {};
							\node [style=nothing] (8) at (2, 0) {};
							\node [style=nothing] (9) at (2, .5) {};
							\end{pgfonlayer}
							\begin{pgfonlayer}{edgelayer}
							\draw [style=simple] (0) to (8);
							\draw [style=simple] (1) to (9);
							\draw [style=simple] (2) to (3);
							\draw [style=simple] (4) to (5);
							\draw [style=simple] (6) to (7);
							\end{pgfonlayer}
							\end{tikzpicture}
							=
							\begin{tikzpicture}
	\begin{pgfonlayer}{nodelayer}
		\node [style=nothing] (0) at (0, -0) {};
		\node [style=nothing] (1) at (0, 0.5000002) {};
		\node [style=nothing] (2) at (1, 0.5000002) {};
		\node [style=nothing] (3) at (1, -0) {};
	\end{pgfonlayer}
	\begin{pgfonlayer}{edgelayer}
		\draw [in=180, out=0, looseness=1.25] (1) to (3);
		\draw [in=180, out=0, looseness=1.25] (0) to (2);
	\end{pgfonlayer}
\end{tikzpicture}$}

						\item
						\label{CNOT.2}
						\hfil{
							$
							\begin{tikzpicture}
							\begin{pgfonlayer}{nodelayer}
							\node [style=nothing] (0) at (0, 0) {};
							\node [style=nothing] (1) at (0, .5) {};
							\node [style=oplus] (2) at (.5, 0) {};
							\node [style=dot] (3) at (.5, .5) {};
							\node [style=oplus] (6) at (1, 0) {};
							\node [style=dot] (7) at (1, .5) {};
							\node [style=nothing] (8) at (1.5, 0) {};
							\node [style=nothing] (9) at (1.5, .5) {};
							\end{pgfonlayer}
							\begin{pgfonlayer}{edgelayer}
							\draw [style=simple] (0) to (8);
							\draw [style=simple] (1) to (9);
							\draw [style=simple] (2) to (3);
							\draw [style=simple] (6) to (7);
							\end{pgfonlayer}
							\end{tikzpicture}
							=
							\begin{tikzpicture}
							\begin{pgfonlayer}{nodelayer}
							\node [style=nothing] (0) at (0, 0) {};
							\node [style=nothing] (1) at (0, .5) {};
							\node [style=nothing] (3) at (1.5, 0) {};
							\node [style=nothing] (4) at (1.5, .5) {};
							\end{pgfonlayer}
							\begin{pgfonlayer}{edgelayer}
							\draw [style=simple] (0) to (3);
							\draw [style=simple] (1) to (4);
							\end{pgfonlayer}
							\end{tikzpicture}
							$}
						
						\item
						\label{CNOT.3}
						\hfil{
							$
							\begin{tikzpicture}
							\begin{pgfonlayer}{nodelayer}
							\node [style=nothing] (0) at (0, 1) {};
							\node [style=nothing] (1) at (0, .5) {};
							\node [style=nothing] (2) at (0, 0) {};
							\node [style=oplus] (3) at (.75, 1) {};
							\node [style=dot] (4) at (.75, .5) {};
							\node [style=dot] (5) at (1.25, .5) {};
							\node [style=oplus] (6) at (1.25, 0) {};
							\node [style=nothing] (7) at (2, 1) {};
							\node [style=nothing] (8) at (2, .5) {};
							\node [style=nothing] (9) at (2, 0) {};
							\end{pgfonlayer}
							\begin{pgfonlayer}{edgelayer}
							\draw [style=simple] (0) to (7);
							\draw [style=simple] (1) to (8);
							\draw [style=simple] (2) to (9);
							\draw [style=simple] (3) to (4);
							\draw [style=simple] (5) to (6);
							\end{pgfonlayer}
							\end{tikzpicture}
							=
							\begin{tikzpicture}
							\begin{pgfonlayer}{nodelayer}
							\node [style=nothing] (0) at (0, 1) {};
							\node [style=nothing] (1) at (0, .5) {};
							\node [style=nothing] (2) at (0, 0) {};
							\node [style=oplus] (3) at (1.25, 1) {};
							\node [style=dot] (4) at (1.25, .5) {};
							\node [style=dot] (5) at (.75, .5) {};
							\node [style=oplus] (6) at (.75, 0) {};
							\node [style=nothing] (7) at (2, 1) {};
							\node [style=nothing] (8) at (2, .5) {};
							\node [style=nothing] (9) at (2, 0) {};
							\end{pgfonlayer}
							\begin{pgfonlayer}{edgelayer}
							\draw [style=simple] (0) to (7);
							\draw [style=simple] (1) to (8);
							\draw [style=simple] (2) to (9);
							\draw [style=simple] (3) to (4);
							\draw [style=simple] (5) to (6);
							\end{pgfonlayer}
							\end{tikzpicture}
							$}
						
						\item 
						\label{CNOT.4}
						\hfil{
							$
							\begin{tikzpicture}
							\begin{pgfonlayer}{nodelayer}
							\node [style=onein] (0) at (0, .5) {};
							\node [style=nothing] (1) at (0, 0) {};
							\node [style=dot] (2) at (.5, .5) {};
							\node [style=oplus] (3) at (.5, 0) {};
							\node [style=nothing] (4) at (1, .5) {};
							\node [style=nothing] (5) at (1, 0) {};
							\end{pgfonlayer}
							\begin{pgfonlayer}{edgelayer}
							\draw [style=simple] (0) to (4);
							\draw [style=simple] (1) to (5);
							\draw [style=simple] (2) to (3);
							\end{pgfonlayer}
							\end{tikzpicture}
							=
							\begin{tikzpicture}
							\begin{pgfonlayer}{nodelayer}
							\node [style=onein] (0) at (0, .5) {};
							\node [style=nothing] (1) at (0, 0) {};
							\node [style=dot] (2) at (.5, .5) {};
							\node [style=oplus] (3) at (.5, 0) {};
							\node [style=oneout] (4) at (1, .5) {};
							\node [style=nothing] (5) at (2, 0) {};
							\node [style=onein] (6) at (1.5, .5) {};
							\node [style=nothing] (7) at (2, 0.5) {};
							\end{pgfonlayer}
							\begin{pgfonlayer}{edgelayer}
							\draw [style=simple] (0) to (4);
							\draw [style=simple] (1) to (5);
							\draw [style=simple] (2) to (3);
							\draw [style=simple] (6) to (7);
							\end{pgfonlayer}
							\end{tikzpicture}
							\text{, }
							\begin{tikzpicture}
							\begin{pgfonlayer}{nodelayer}
							\node [style=nothing] (0) at (0, .5) {};
							\node [style=nothing] (1) at (0, 0) {};
							\node [style=dot] (2) at (.5, .5) {};
							\node [style=oplus] (3) at (.5, 0) {};
							\node [style=oneout] (4) at (1, .5) {};
							\node [style=nothing] (5) at (1, 0) {};
							\end{pgfonlayer}
							\begin{pgfonlayer}{edgelayer}
							\draw [style=simple] (0) to (4);
							\draw [style=simple] (1) to (5);
							\draw [style=simple] (2) to (3);
							\end{pgfonlayer}
							\end{tikzpicture}
							=
							\begin{tikzpicture}
							\begin{pgfonlayer}{nodelayer}
							\node [style=oneout] (0) at (2, .5) {};
							\node [style=nothing] (1) at (2, 0) {};
							\node [style=dot] (2) at (1.5, .5) {};
							\node [style=oplus] (3) at (1.5, 0) {};
							\node [style=onein] (4) at (1, .5) {};
							\node [style=nothing] (5) at (0, 0) {};
							\node [style=oneout] (6) at (.5, .5) {};
							\node [style=nothing] (7) at (0, 0.5) {};
							\end{pgfonlayer}
							\begin{pgfonlayer}{edgelayer}
							\draw [style=simple] (0) to (4);
							\draw [style=simple] (1) to (5);
							\draw [style=simple] (2) to (3);
							\draw [style=simple] (6) to (7);
							\end{pgfonlayer}
							\end{tikzpicture}
							$}
						
						\item 
						\label{CNOT.5}
						\hfil{
							$
							\begin{tikzpicture}
							\begin{pgfonlayer}{nodelayer}
							\node [style=nothing] (0) at (0, 1) {};
							\node [style=nothing] (1) at (0, .5) {};
							\node [style=nothing] (2) at (0, 0) {};
							\node [style=dot] (3) at (.75, 1) {};
							\node [style=oplus] (4) at (.75, .5) {};
							\node [style=oplus] (5) at (1.25, .5) {};
							\node [style=dot] (6) at (1.25, 0) {};
							\node [style=nothing] (7) at (2, 1) {};
							\node [style=nothing] (8) at (2, .5) {};
							\node [style=nothing] (9) at (2, 0) {};
							\end{pgfonlayer}
							\begin{pgfonlayer}{edgelayer}
							\draw [style=simple] (0) to (7);
							\draw [style=simple] (1) to (8);
							\draw [style=simple] (2) to (9);
							\draw [style=simple] (3) to (4);
							\draw [style=simple] (5) to (6);
							\end{pgfonlayer}
							\end{tikzpicture}
							=
							\begin{tikzpicture}
							\begin{pgfonlayer}{nodelayer}
							\node [style=nothing] (0) at (0, 1) {};
							\node [style=nothing] (1) at (0, .5) {};
							\node [style=nothing] (2) at (0, 0) {};
							\node [style=dot] (3) at (1.25, 1) {};
							\node [style=oplus] (4) at (1.25, .5) {};
							\node [style=oplus] (5) at (.75, .5) {};
							\node [style=dot] (6) at (.75, 0) {};
							\node [style=nothing] (7) at (2, 1) {};
							\node [style=nothing] (8) at (2, .5) {};
							\node [style=nothing] (9) at (2, 0) {};
							\end{pgfonlayer}
							\begin{pgfonlayer}{edgelayer}
							\draw [style=simple] (0) to (7);
							\draw [style=simple] (1) to (8);
							\draw [style=simple] (2) to (9);
							\draw [style=simple] (3) to (4);
							\draw [style=simple] (5) to (6);
							\end{pgfonlayer}
							\end{tikzpicture}
							$}
						
						\item 
						\label{CNOT.6}
						\hfil{
							$
							\begin{tikzpicture}
							\begin{pgfonlayer}{nodelayer}
							\node [style=onein] (0) at (0, 0) {};
							\node [style=oneout] (1) at (1, 0) {};
							\end{pgfonlayer}
							\begin{pgfonlayer}{edgelayer}
							\draw [style=simple] (0) to (1);
							\end{pgfonlayer}
							\end{tikzpicture}
							=
							1_0
							$}
						
						\item 
						\label{CNOT.7}
						\hfil{
							$
							\begin{tikzpicture}
							\begin{pgfonlayer}{nodelayer}
							\node [style=onein] (0) at (0, 1) {};
							\node [style=onein] (1) at (0, .5) {};
							\node [style=nothing] (2) at (0, 0) {};
							\node [style=dot] (3) at (.5, 1) {};
							\node [style=oplus] (4) at (.5, .5) {};
							\node [style=dot] (5) at (1, .5) {};
							\node [style=oplus] (6) at (1, 0) {};
							\node [style=oneout] (7) at (1, 1) {};
							\node [style=nothing] (8) at (1.5, .5) {};
							\node [style=nothing] (9) at (1.5, 0) {};
							\end{pgfonlayer}
							\begin{pgfonlayer}{edgelayer}
							\draw [style=simple] (0) to (7);
							\draw [style=simple] (1) to (8);
							\draw [style=simple] (2) to (9);
							\draw [style=simple] (3) to (4);
							\draw [style=simple] (5) to (6);
							\end{pgfonlayer}
							\end{tikzpicture}
							=
							\begin{tikzpicture}
							\begin{pgfonlayer}{nodelayer}
							\node [style=onein] (0) at (0, 1) {};
							\node [style=onein] (1) at (0, .5) {};
							\node [style=nothing] (2) at (0, 0) {};
							\node [style=dot] (3) at (.5, 1) {};
							\node [style=oplus] (4) at (.5, .5) {};
							\node [style=oneout] (7) at (1, 1) {};
							\node [style=nothing] (8) at (1.5, .5) {};
							\node [style=nothing] (9) at (1.5, 0) {};
							\end{pgfonlayer}
							\begin{pgfonlayer}{edgelayer}
							\draw [style=simple] (0) to (7);
							\draw [style=simple] (1) to (8);
							\draw [style=simple] (2) to (9);
							\draw [style=simple] (3) to (4);
							\end{pgfonlayer}
							\end{tikzpicture}
							\text{, }
							\begin{tikzpicture}
							\begin{pgfonlayer}{nodelayer}
							\node [style=oneout] (0) at (1.5, 1) {};
							\node [style=oneout] (1) at (1.5, .5) {};
							\node [style=nothing] (2) at (1.5, 0) {};
							\node [style=dot] (3) at (1, 1) {};
							\node [style=oplus] (4) at (1, .5) {};
							\node [style=dot] (5) at (.5, .5) {};
							\node [style=oplus] (6) at (.5, 0) {};
							\node [style=onein] (7) at (.5, 1) {};
							\node [style=nothing] (8) at (0, .5) {};
							\node [style=nothing] (9) at (0, 0) {};
							\end{pgfonlayer}
							\begin{pgfonlayer}{edgelayer}
							\draw [style=simple] (0) to (7);
							\draw [style=simple] (1) to (8);
							\draw [style=simple] (2) to (9);
							\draw [style=simple] (3) to (4);
							\draw [style=simple] (5) to (6);
							\end{pgfonlayer}
							\end{tikzpicture}
							=
							\begin{tikzpicture}
							\begin{pgfonlayer}{nodelayer}
							\node [style=oneout] (0) at (1.5, 1) {};
							\node [style=oneout] (1) at (1.5, .5) {};
							\node [style=nothing] (2) at (1.5, 0) {};
							\node [style=dot] (3) at (1, 1) {};
							\node [style=oplus] (4) at (1, .5) {};
							\node [style=onein] (7) at (.5, 1) {};
							\node [style=nothing] (8) at (0, .5) {};
							\node [style=nothing] (9) at (0, 0) {};
							\end{pgfonlayer}
							\begin{pgfonlayer}{edgelayer}
							\draw [style=simple] (0) to (7);
							\draw [style=simple] (1) to (8);
							\draw [style=simple] (2) to (9);
							\draw [style=simple] (3) to (4);
							\end{pgfonlayer}
							\end{tikzpicture}
							$}
						
						\item 
						\label{CNOT.8}
						\hfil{
							$
							\begin{tikzpicture}
							\begin{pgfonlayer}{nodelayer}
							\node [style=nothing] (0) at (0, 1) {};
							\node [style=nothing] (1) at (0, .5) {};
							\node [style=nothing] (2) at (0, 0) {};
							\node [style=dot] (3) at (.5, 1) {};
							\node [style=oplus] (4) at (.5, .5) {};
							\node [style=dot] (5) at (1, .5) {};
							\node [style=oplus] (6) at (1, 0) {};
							\node [style=dot] (7) at (1.5, 1) {};
							\node [style=oplus] (8) at (1.5, .5) {};
							\node [style=nothing] (9) at (2, 1) {};
							\node [style=nothing] (10) at (2, .5) {};
							\node [style=nothing] (11) at (2, 0) {};
							\end{pgfonlayer}
							\begin{pgfonlayer}{edgelayer}
							\draw [style=simple] (0) to (9);
							\draw [style=simple] (1) to (10);
							\draw [style=simple] (2) to (11);
							\draw [style=simple] (3) to (4);
							\draw [style=simple] (5) to (6);
							\draw [style=simple] (7) to (8);
							\end{pgfonlayer}
							\end{tikzpicture}
							=
							\begin{tikzpicture}
							\begin{pgfonlayer}{nodelayer}
							\node [style=nothing] (0) at (0, 1) {};
							\node [style=nothing] (1) at (0, .5) {};
							\node [style=nothing] (2) at (0, 0) {};
							\node [style=dot] (5) at (.5, .5) {};
							\node [style=oplus] (6) at (.5, 0) {};
							\node [style=dot] (7) at (1, 1) {};
							\node [style=oplus] (8) at (1, 0) {};
							\node [style=nothing] (9) at (1.5, 1) {};
							\node [style=nothing] (10) at (1.5, .5) {};
							\node [style=nothing] (11) at (1.5, 0) {};
							\end{pgfonlayer}
							\begin{pgfonlayer}{edgelayer}
							\draw [style=simple] (0) to (9);
							\draw [style=simple] (1) to (10);
							\draw [style=simple] (2) to (11);
							\draw [style=simple] (5) to (6);
							\draw [style=simple] (7) to (8);
							\end{pgfonlayer}
							\end{tikzpicture}
							$}
						
						\item 
						\label{CNOT.9}
						\hfil{
							$
							\begin{tikzpicture}
							\begin{pgfonlayer}{nodelayer}
							\node [style=onein] (0) at (0, 1) {};
							\node [style=onein] (1) at (0, .5) {};
							\node [style=nothing] (2) at (0, 0) {};
							\node [style=dot] (3) at (.5, 1) {};
							\node [style=oplus] (4) at (.5, .5) {};
							\node [style=oneout] (7) at (1, 1) {};
							\node [style=oneout] (8) at (1, .5) {};
							\node [style=nothing] (9) at (1, 0) {};
							\end{pgfonlayer}
							\begin{pgfonlayer}{edgelayer}
							\draw [style=simple] (0) to (7);
							\draw [style=simple] (1) to (8);
							\draw [style=simple] (2) to (9);
							\draw [style=simple] (3) to (4);
							\end{pgfonlayer}
							\end{tikzpicture}
							=
							\begin{tikzpicture}
							\begin{pgfonlayer}{nodelayer}
							\node [style=onein] (0) at (0, 1) {};
							\node [style=onein] (1) at (0, .5) {};
							\node [style=nothing] (2) at (0, 0) {};
							\node [style=dot] (3) at (1, 1) {};
							\node [style=oplus] (4) at (1, .5) {};
							\node [style=oneout] (7) at (2, 1) {};
							\node [style=oneout] (8) at (2, .5) {};
							\node [style=nothing] (9) at (2, 0) {};
							\node [style=oneout] (10) at (0.75, 0) {};
							\node [style=onein] (11) at (1.25, 0) {};
							\end{pgfonlayer}
							\begin{pgfonlayer}{edgelayer}
							\draw [style=simple] (0) to (7);
							\draw [style=simple] (1) to (8);
							\draw [style=simple] (2) to (10);
							\draw [style=simple] (11) to (9);
							\draw [style=simple] (3) to (4);
							\end{pgfonlayer}
							\end{tikzpicture}
							$}
					\end{enumerate}
				\end{multicols}
				\
			\end{mdframed}
	}}
	\caption{The identities of $\CNOT$}
	\label{fig:CNOT}
\end{figure}
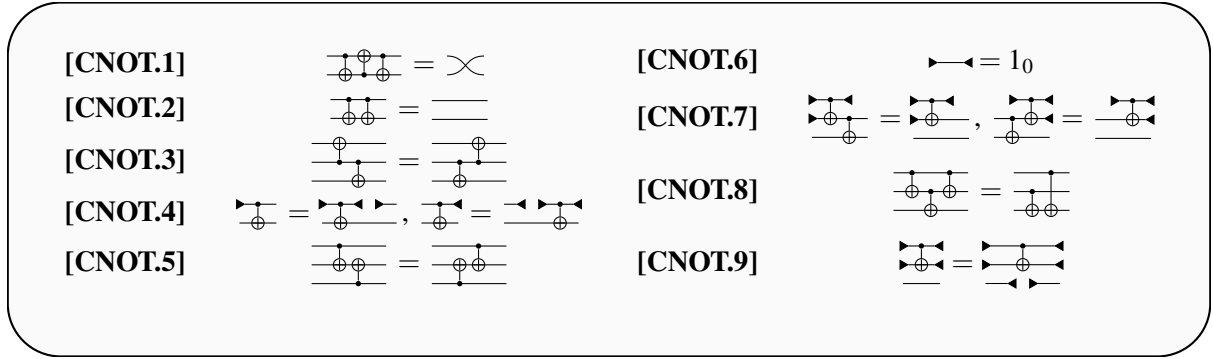

Notice that there are ``gaps'' in some of the $\cnot$ gates, and others are flipped.  This is just shorthand to suppress burdensome symmetry maps. There is an obvious interpretation of this notation, for example:
$$
\begin{tikzpicture}
	\begin{pgfonlayer}{nodelayer}
		\node [style=nothing] (0) at (0, 1.5) {};
		\node [style=nothing] (1) at (0, 2.5) {};
		\node [style=dot] (2) at (0.5000001, 2.5) {};
		\node [style=oplus] (3) at (0.5000001, 1.5) {};
		\node [style=nothing] (4) at (1, 2.5) {};
		\node [style=nothing] (5) at (1, 1.5) {};
		\node [style=nothing] (6) at (0, 2) {};
		\node [style=nothing] (7) at (1, 2) {};
	\end{pgfonlayer}
	\begin{pgfonlayer}{edgelayer}
		\draw (1) to (2);
		\draw (3) to (0);
		\draw (3) to (2);
		\draw (6) to (7);
		\draw (2) to (4);
		\draw (5) to (3);
	\end{pgfonlayer}
\end{tikzpicture}
:=
\begin{tikzpicture}
	\begin{pgfonlayer}{nodelayer}
		\node [style=nothing] (0) at (-0.25, 1.5) {};
		\node [style=nothing] (1) at (-0.25, 2.5) {};
		\node [style=dot] (2) at (0.5000001, 2.5) {};
		\node [style=oplus] (3) at (0.5000001, 2) {};
		\node [style=nothing] (4) at (1.25, 2.5) {};
		\node [style=nothing] (5) at (1.25, 1.5) {};
		\node [style=nothing] (6) at (-0.25, 2) {};
		\node [style=nothing] (7) at (0.5000001, 1.5) {};
		\node [style=nothing] (8) at (1.25, 2) {};
	\end{pgfonlayer}
	\begin{pgfonlayer}{edgelayer}
		\draw (1) to (2);
		\draw [in=0, out=180, looseness=1.00] (3) to (0);
		\draw (3) to (2);
		\draw [in=180, out=0, looseness=1.00] (6) to (7);
		\draw (2) to (4);
		\draw [in=0, out=180, looseness=1.00] (5) to (3);
		\draw [in=180, out=0, looseness=1.00] (7) to (8);
	\end{pgfonlayer}
\end{tikzpicture}
$$

It was shown that these identities were complete;  and moreover, it was shown that $\CNOT$ is equivalent to the category of affine partial isomorphisms between finite $\Z_2$-vector spaces.   We shall use these observations to obtain a similar equivalence for the category, $\TOF$, generated by the Toffoli gate and 1 ancillary bits.

\section{The category \texorpdfstring{$\TOF$}{TOF}}

Define the category $\TOF$ to be the symmetric monoidal category, with objects the natural numbers, generated by the 1 ancillary bits $\onein$ and $\oneout$ (depicted graphically as in $\CNOT$) as well as the Toffoli gate:   
  \[  \tof := 

		\]
The Toffoli gate and the 1-ancillary bits allow the $\cnot$ gate, $\Not$ gate, $\zeroin$ gate, $\zeroout$ gate, flipped $\tof$ gate and flipped $\cnot$ gate to be succinctly defined:
\[ %
  \]

We also allow for ``gaps'' in $\tof$ and $\cnot$ gates, as in $\CNOT$.

These gates must satisfy the identities given in Figure \ref{fig:TOF}:

\begin{figure}[h]
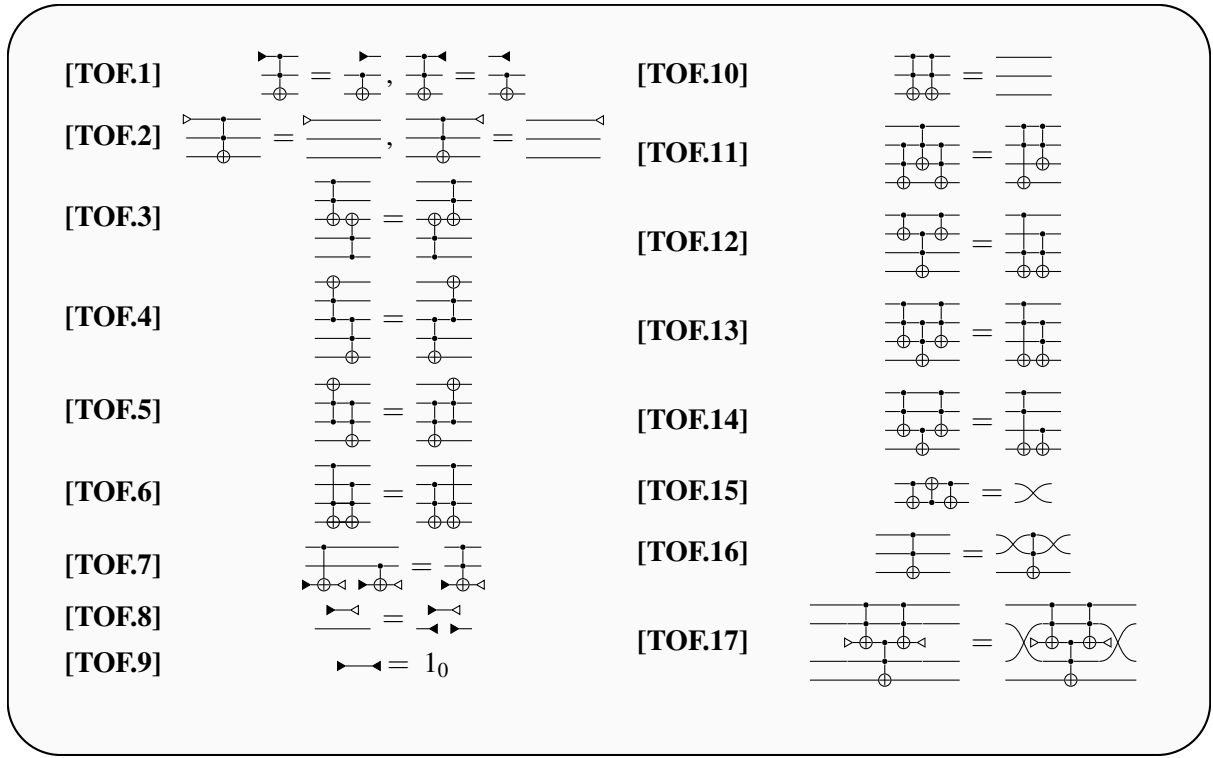

\noindent
\scalebox{2.0}{%
\vbox{%
\begin{mdframed}[backgroundcolor=black!2,roundcorner=10.0pt, innerleftmargin=+.75cm]
\begin{multicols}{2}
\begin{enumerate}[label={\bf [TOF.\arabic*]}, ref={\bf [TOF.\arabic*]}, wide = 0pt, leftmargin = 2em]
\item
\label{TOF.1}
{\hfil
$
%
$}
\end{enumerate}
\end{multicols}
\
\end{mdframed}
}}
\caption{The identities of $\TOF$}
\label{fig:TOF}
\end{figure}

Axioms \ref{TOF.1}-\ref{TOF.5.5}, \ref{TOF.8}-\ref{TOF.9} are relatively intuitive.  \ref{TOF.7} corresponds to tensoring a matrix with the $1\times 1$ zero matrix: this is useful for establishing the restriction structure of $\TOF$. \ref{TOF.6} is used for establishing a normal form for the restriction idempotents.  Any of Axioms \ref{TOF.10}-\ref{TOF.13} combined with \ref{TOF.9} are used to push $\cnot$/$\tof$ gates past each other thereby generating another trailing $\cnot$/$\tof$ gate: this will be discussed in more detail in Lemma \ref{remark:Iwama} {\em (v)} .  \ref{TOF.14} is inherited from $\CNOT$ and is used to establish the inverse products of $\TOF$.  \ref{TOF.15} expresses the commutativity of multiplication for the Toffoli gate.  \ref{TOF.16} is similar to \ref{TOF.15}, except for the 3-bit-controlled-not gate.

As an exercise, we first show that $\zeroout\circ\zeroin = 1_0$ (just as $\oneout\circ \onein=1_0$ in \ref{TOF.8}):

$$
\begin{tikzpicture}
\begin{pgfonlayer}{nodelayer}
\node [style=zeroin] (0) at (-3, -4) {};
\node [style=zeroout] (1) at (-2, -4) {};
\end{pgfonlayer}
\begin{pgfonlayer}{edgelayer}
\draw [style=simple] (0) to (1);
\end{pgfonlayer}
\end{tikzpicture}
=
\begin{tikzpicture}
\begin{pgfonlayer}{nodelayer}
\node [style=oplus] (0) at (-3, -5) {};
\node [style=dot] (1) at (-3, -4.5) {};
\node [style=oplus] (2) at (-1.5, -5) {};
\node [style=oneout] (3) at (-1, -5) {};
\node [style=onein] (4) at (-3.5, -4) {};
\node [style=oneout] (5) at (-1, -4.5) {};
\node [style=onein] (6) at (-2, -4.5) {};
\node [style=dot] (7) at (-1.5, -4.5) {};
\node [style=oneout] (8) at (-2.5, -4.5) {};
\node [style=dot] (9) at (-1.5, -4) {};
\node [style=onein] (10) at (-3.5, -5) {};
\node [style=dot] (11) at (-3, -4) {};
\node [style=onein] (12) at (-3.5, -4.5) {};
\node [style=oneout] (13) at (-1, -4) {};
\end{pgfonlayer}
\begin{pgfonlayer}{edgelayer}
\draw [style=simple] (1) to (0);
\draw [style=simple] (10) to (0);
\draw [style=simple] (0) to (2);
\draw [style=simple] (2) to (3);
\draw [style=simple] (5) to (7);
\draw [style=simple] (1) to (12);
\draw [style=simple] (2) to (7);
\draw [style=simple] (8) to (1);
\draw [style=simple] (7) to (6);
\draw [style=simple] (4) to (11);
\draw [style=simple] (11) to (1);
\draw [style=simple] (7) to (9);
\draw [style=simple] (9) to (13);
\draw [style=simple] (11) to (9);
\end{pgfonlayer}
\end{tikzpicture}
=
\begin{tikzpicture}
\begin{pgfonlayer}{nodelayer}
\node [style=oplus] (0) at (-0.25, -5) {};
\node [style=oneout] (1) at (0.25, -5) {};
\node [style=oneout] (2) at (0.5, -4) {};
\node [style=dot] (3) at (-0.25, -4.5) {};
\node [style=dot] (4) at (-0.25, -4) {};
\node [style=oneout] (5) at (0.5, -4.5) {};
\node [style=onein] (6) at (-2.75, -5) {};
\node [style=onein] (7) at (-3, -4.5) {};
\node [style=onein] (8) at (-3, -4) {};
\node [style=dot] (9) at (-2.25, -4.5) {};
\node [style=oplus] (10) at (-2.25, -5) {};
\node [style=dot] (11) at (-2.25, -4) {};
\node [style=onein] (12) at (-1, -4.5) {};
\node [style=oneout] (13) at (-1.5, -4.5) {};
\node [style=nothing] (14) at (-1.25, -4) {};
\end{pgfonlayer}
\begin{pgfonlayer}{edgelayer}
\draw [style=simple] (0) to (1);
\draw [style=simple, in=0, out=180, looseness=1.00] (2) to (3);
\draw [style=simple] (0) to (3);
\draw [style=simple] (3) to (4);
\draw [style=simple, in=180, out=0, looseness=1.00] (4) to (5);
\draw [style=simple] (9) to (10);
\draw [style=simple] (6) to (10);
\draw [style=simple, in=0, out=180, looseness=1.00] (9) to (8);
\draw [style=simple, in=180, out=0, looseness=1.00] (7) to (11);
\draw [style=simple] (11) to (9);
\draw [style=simple] (10) to (0);
\draw [style=simple, in=180, out=0, looseness=1.00] (12) to (4);
\draw [style=simple, in=0, out=180, looseness=1.00] (3) to (14);
\draw [style=simple, in=0, out=180, looseness=1.00] (14) to (9);
\draw [style=simple, in=0, out=180, looseness=1.00] (13) to (11);
\end{pgfonlayer}
\end{tikzpicture}
=
\begin{tikzpicture}
\begin{pgfonlayer}{nodelayer}
\node [style=oplus] (0) at (-0.75, -5) {};
\node [style=oneout] (1) at (-0.25, -5) {};
\node [style=oneout] (2) at (-0.25, -4.5) {};
\node [style=dot] (3) at (-0.75, -4.5) {};
\node [style=dot] (4) at (-0.75, -4) {};
\node [style=oneout] (5) at (-0.25, -4) {};
\node [style=onein] (6) at (-2.75, -5) {};
\node [style=onein] (7) at (-2.75, -4) {};
\node [style=onein] (8) at (-2.75, -4.5) {};
\node [style=dot] (9) at (-2.25, -4.5) {};
\node [style=oplus] (10) at (-2.25, -5) {};
\node [style=dot] (11) at (-2.25, -4) {};
\node [style=onein] (12) at (-1.25, -4) {};
\node [style=oneout] (13) at (-1.75, -4) {};
\end{pgfonlayer}
\begin{pgfonlayer}{edgelayer}
\draw [style=simple] (0) to (1);
\draw [style=simple, in=0, out=180, looseness=1.00] (2) to (3);
\draw [style=simple] (0) to (3);
\draw [style=simple] (3) to (4);
\draw [style=simple, in=180, out=0, looseness=1.00] (4) to (5);
\draw [style=simple] (9) to (10);
\draw [style=simple] (6) to (10);
\draw [style=simple, in=0, out=180, looseness=1.00] (9) to (8);
\draw [style=simple, in=180, out=0, looseness=1.00] (7) to (11);
\draw [style=simple] (11) to (9);
\draw [style=simple] (10) to (0);
\draw [style=simple] (9) to (3);
\draw [style=simple] (4) to (12);
\draw [style=simple] (13) to (11);
\end{pgfonlayer}
\end{tikzpicture}
=
\begin{tikzpicture}
\begin{pgfonlayer}{nodelayer}
\node [style=onein] (0) at (-3.5, -5) {};
\node [style=oneout] (1) at (-1, -4.5) {};
\node [style=onein] (2) at (-3.5, -4) {};
\node [style=dot] (3) at (-1.5, -4) {};
\node [style=dot] (4) at (-1.5, -4.5) {};
\node [style=onein] (5) at (-3.5, -4.5) {};
\node [style=oneout] (6) at (-2.5, -4) {};
\node [style=dot] (7) at (-3, -4) {};
\node [style=oneout] (8) at (-1, -4) {};
\node [style=oplus] (9) at (-3, -5) {};
\node [style=onein] (10) at (-2, -4.5) {};
\node [style=dot] (11) at (-3, -4.5) {};
\node [style=oplus] (12) at (-1.5, -5) {};
\node [style=oneout] (13) at (-1, -5) {};
\node [style=onein] (14) at (-2, -4) {};
\node [style=oneout] (15) at (-2.5, -4.5) {};
\end{pgfonlayer}
\begin{pgfonlayer}{edgelayer}
\draw [style=simple] (11) to (9);
\draw [style=simple] (0) to (9);
\draw [style=simple] (9) to (12);
\draw [style=simple] (12) to (13);
\draw [style=simple] (1) to (4);
\draw [style=simple] (11) to (5);
\draw [style=simple] (12) to (4);
\draw [style=simple] (15) to (11);
\draw [style=simple] (4) to (10);
\draw [style=simple] (2) to (7);
\draw [style=simple] (7) to (6);
\draw [style=simple] (7) to (11);
\draw [style=simple] (4) to (3);
\draw [style=simple] (3) to (8);
\draw [style=simple] (3) to (14);
\end{pgfonlayer}
\end{tikzpicture}
=
\begin{tikzpicture}
\begin{pgfonlayer}{nodelayer}
\node [style=onein] (0) at (-3.5, -5) {};
\node [style=oneout] (1) at (-2, -5) {};
\node [style=oneout] (2) at (-2, -4.5) {};
\node [style=onein] (3) at (-3.5, -4.5) {};
\node [style=dot] (4) at (-3, -4.5) {};
\node [style=dot] (5) at (-2.5, -4.5) {};
\node [style=oplus] (6) at (-3, -5) {};
\node [style=oplus] (7) at (-2.5, -5) {};
\node [style=dot] (8) at (-3, -4) {};
\node [style=dot] (9) at (-2.5, -4) {};
\node [style=onein] (10) at (-3.5, -4) {};
\node [style=oneout] (11) at (-2, -4) {};
\end{pgfonlayer}
\begin{pgfonlayer}{edgelayer}
\draw [style=simple] (4) to (6);
\draw [style=simple] (0) to (6);
\draw [style=simple] (6) to (7);
\draw [style=simple] (7) to (1);
\draw [style=simple] (2) to (5);
\draw [style=simple] (4) to (3);
\draw [style=simple] (7) to (5);
\draw [style=simple] (10) to (8);
\draw [style=simple] (8) to (4);
\draw [style=simple] (5) to (9);
\draw [style=simple] (9) to (11);
\draw [style=simple] (4) to (5);
\draw [style=simple] (9) to (8);
\end{pgfonlayer}
\end{tikzpicture}
=
\begin{tikzpicture}
\begin{pgfonlayer}{nodelayer}
\node [style=onein] (0) at (-3.5, -5) {};
\node [style=oneout] (1) at (-3, -5) {};
\node [style=oneout] (2) at (-3, -4.5) {};
\node [style=onein] (3) at (-3.5, -4.5) {};
\node [style=onein] (4) at (-3.5, -4) {};
\node [style=oneout] (5) at (-3, -4) {};
\end{pgfonlayer}
\begin{pgfonlayer}{edgelayer}
\draw [style=simple] (0) to (1);
\draw [style=simple] (2) to (3);
\draw [style=simple] (4) to (5);
\end{pgfonlayer}
\end{tikzpicture}
=
1_0
$$

Since we build on the work done in \cite{CNOT}, we start by establishing that there is a canonical functor from the category $\CNOT$ into $\TOF$.  The proof is contained in Appendix  \ref{appendix:functorial}.
\begin{lemma}
\label{lemma:functorial}
The canonical interpretation of $\CNOT$ in $\TOF$ is functorial.
\end{lemma}


\section{Controlled not gates and the Iwama identities}


In reversible and quantum computing it is usual to regard the Toffoli gate and $\cnot$ as the ${\sf not}$ gate ``controlled'', respectively, by one and two control wires.  In $\TOF$ 
one can define $\cnot_n$ the ${\sf not}$ gate controlled by $n$ wires for any $n \in \mathbb{N}$:

\begin{definition}
\label{definition:Generalizedcnot}
For every $n\in\N$, inductively define the controlled not gate, $\cnot_n:n+1\to n+1$ inductively by:
\begin{itemize}
	\item
	For the base cases, let $\cnot_0:=\notgate$, $\cnot_1:=\cnot$ and $\cnot_2:=\tof$.

	\item
	For all $n \in \N$ such that $n\geq 2$:
	
	\[
	\cnot_{n+1} \equiv

	\]	
	
\end{itemize}
\end{definition}

The first $n$ wires of $\cnot_n$ are the {\bf control wires} while the last wire is the {\bf target wire}.  As for the $\cnot$ and $\tof$ gates, allow for $\cnot_n$ gates to haves ``gaps'' in between the control wires (and the target wire).

Notice that $\cnot_n$ gates are self inverse:

\begin{lemma}
	\label{lemma:PolySelfInverse}
	Every $\cnot_n$ gate is self-inverse.
\end{lemma}

\begin{proof}
	The base cases are trivial.  For the inductive step, suppose that $\cnot_n$ gates are self-inverse, then:

	\begin{align*}
	%
	\end{align*}

\end{proof}

Next we show how we can ``unzip'' $\cnot_n$ gates; and simultaneously, we generalize Lemma \ref{lemma:pushingGates} to show how a $\cnot_n$ gate can be ``pushed'' past a Toffoli gate.  This is 
a crucial observation in establishing the Iwama identities.  The proof is contained in Appendix \ref{appendix:zipper}.

\begin{proposition} Given some $n\geq 1$ and $k\geq 1$:
	\label{lemma:zipper}
\begin{enumerate}[label={\em (\roman*)}]
	\item $\cnot_{n+k}$ gates can be zipped and unzipped:
\[
%
\]
\item $\cnot_n$ gates can be pushed past Toffoli gates in the following sense:
	\[
%
\end{align*}

\end{enumerate}
\end{proposition}

We can use the ability to zip and unzip controlled not gates to show that we can transpose control wires:

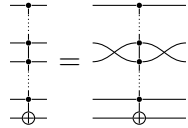
\begin{corollary}
\label{cor:transpose}
	$$
	%
	$$
\end{corollary}

\begin{proof}
The claim follows vacuously for $n<2$.  The base case when $n=2$ follows immediately from \ref{TOF.15}.  Suppose that the claim holds for some $\cnot_n$.  Consider a $\cnot_{n+1}$ gate and unzip this gate one level down.  We can transpose the top two control wires by the base case and we can transpose the bottom $n-1$ control wires by the inductive hypothesis.  Finally, to transpose the second and third wire from the top, either unzip the $\cnot_{n+1}$ gate down by one more step.  If this cannot be done -- because there is nothing more to unzip -- it means one can directly apply {\bf [TOF.17]}.
\end{proof}

Since transpositions generate the symmetric group, the control wires of a $\cnot_n$ gate can, therefore, be freely permuted.

We shall the notation $\oa_x^X$ to denote a $\cnot_{|X|}$ gate operating on the control wires in the set $X$ and targeting the wire $x$.  Similarly let $\rhd_x$ and $\lhd_x$ be, respectively, the  
$\zeroin$ and $\zeroout$ gates on the wire $x$.    

The following identities are due to Iwama et al. \cite{Iwama}.   They show how to push generalized control not gates past each other in certain key situations.  The proof is contained in Appendix \ref{appendix:Iwama}.

\begin{lemma}[Iwama's Identities] ~
 \label{remark:Iwama}
\begin{enumerate}[label={\em (\roman*)}]
\item $\oa_x^X\oa_x^X=1$
\item When $x \in X$ then $\rhd_x \oa_y^{X} = \rhd_x$
\item  When the target wire are the same $\oa_x^X \oa_x^Y = \oa_x^Y \oa_x^X$
\item  When $x \not\in Y$ and $y \not\in X$ then $\oa_x^X \oa_y^Y = \oa_y^Y \oa_x^X$
\item  $\oa_x^X \oa_y^{\{ x\} \sqcup Y} = \oa_y^{X \cup Y} \oa_y^{\{ x\} \sqcup Y} \oa_x^X$
\item $\rhd_z \oa_z^{\{x\}}\oa_y^{\{x\}\sqcup X} = \rhd_z \oa_z^{\{x\}}\oa_y^{\{z\}\sqcup X}$
\end{enumerate}
\end{lemma}


\section{\texorpdfstring{$\TOF$}{TOF} is a discrete inverse category}

Next we prove that $\TOF$ is a discrete inverse category.

Define $(\_)^\cnv:\TOF^\op\to\TOF$ to be the contravariant functor which flips circuits horizontally, mapping $\onein\mapsto\oneout$, $\oneout \mapsto \onein$  and $\tof\mapsto \tof$.  Clearly $(\_)^\cnv$ is an involution (and thus, a dagger functor) as all of the axioms are horizontally symmetric and $(f^\cnv)^\cnv = f$.

The diagonal map in $\TOF$ is the image of the diagonal map, $\Delta: n \to n \otimes n$, in $\CNOT$ under the canonical functor $\CNOT \to \TOF$:

\begin{definition}
	\label{definition:Delta}
	Define a family of maps $\Delta:=\{\Delta_n\}_{n \in \N}$  with the $\cnot$ gate and $1$-ancillary wires inductively, as in \cite{CNOT}, such that 
	$\Delta_0 = 1_0$,
	
	\begin{center}
		$
		\Delta_1
		:=
		\begin{tikzpicture}
		\begin{pgfonlayer}{nodelayer}
		\node [style=fanout] (0) at (0, 0) {};
		\node [style=nothing] (1) at (-1, 0) {};
		\node [style=nothing] (2) at (1, 0.5) {};
		\node [style=nothing] (3) at (1, -0.5) {};
		\end{pgfonlayer}
		\begin{pgfonlayer}{edgelayer}
		\draw [style=simple, in=180, out=27, looseness=1.00] (0) to (2);
		\draw [style=simple, in=-27, out=180, looseness=1.00] (3) to (0);
		\draw [style=simple] (0) to (1);
		\end{pgfonlayer}
		\end{tikzpicture}
		:=
		\begin{tikzpicture}
		\begin{pgfonlayer}{nodelayer}
		\node [style=nothing] (0) at (-1, 0) {};
		\node [style=nothing] (1) at (1, 0) {};
		\node [style=nothing] (2) at (1, -0.5) {};
		\node [style=dot] (3) at (0, 0) {};
		\node [style=oplus] (4) at (0, -0.5) {};
		\node [style=zeroin] (5) at (-0.5, -0.5) {};
		\end{pgfonlayer}
		\begin{pgfonlayer}{edgelayer}
		\draw [style=simple] (0) to (3);
		\draw [style=simple] (3) to (1);
		\draw [style=simple] (2) to (4);
		\draw [style=simple] (4) to (5);
		\draw [style=simple] (4) to (3);
		\end{pgfonlayer}
		\end{tikzpicture}
		$
		\hspace{1cm}
		and for all $n> 1$:
		\hspace{.5cm}
		$\Delta_n
		:=
		\begin{tikzpicture}
		\begin{pgfonlayer}{nodelayer}
		\node [style=fanout] (0) at (0, 0) {};
		\node [style=nothing] (1) at (-1, 0) {};
		\node [style=nothing] (2) at (1, 0.5) {};
		\node [style=nothing] (3) at (1, -0.5) {};
		\node [style=nothing] (4) at (-0.5, 0.25) {$n$};
		\end{pgfonlayer}
		\begin{pgfonlayer}{edgelayer}
		\draw [style=simple, in=180, out=27, looseness=1.00] (0) to (2);
		\draw [style=simple, in=-27, out=180, looseness=1.00] (3) to (0);
		\draw [style=simple] (0) to (1);
		\end{pgfonlayer}
		\end{tikzpicture}
		:=
		\begin{tikzpicture}
		\begin{pgfonlayer}{nodelayer}
		\node [style=fanout] (0) at (0, 0) {};
		\node [style=nothing] (1) at (-1, 0) {};
		\node [style=nothing] (2) at (1, 0.5) {};
		\node [style=nothing] (3) at (1, -0.5) {};
		\node [style=nothing] (4) at (-0.55, 0.25) {$n-1$};
		\node [style=nothing] (5) at (1, -0.25) {};
		\node [style=fanout] (6) at (0, -0.75) {};
		\node [style=nothing] (7) at (-1, -0.75) {};
		\node [style=nothing] (8) at (1, -1.25) {};
		\end{pgfonlayer}
		\begin{pgfonlayer}{edgelayer}
		\draw [style=simple, in=180, out=27, looseness=1.00] (0) to (2);
		\draw [style=simple, in=-27, out=180, looseness=1.00] (3) to (0);
		\draw [style=simple] (0) to (1);
		\draw [style=simple, in=180, out=27, looseness=1.00] (6) to (5);
		\draw [style=simple, in=-27, out=180, looseness=1.00] (8) to (6);
		\draw [style=simple] (6) to (7);
		\end{pgfonlayer}
		\end{tikzpicture}
		$
	\end{center}
\end{definition}

This yields the following result;  the proof is contained in Appendix \ref{appendix:discreteInverse}:

\begin{proposition}
	\label{propositoin:TOFDiscreteInverse}
	$\TOF$ is a discrete inverse category with this structure.
\end{proposition}

\section{The points of \texorpdfstring{$\TOF$}{TOF}}

The (total) points of $\TOF$ are iterated tensors of the input ancillary bits ($\zeroin$ and $\onein$) and can be denoted using the ``ket'' notation.  We start by observing  that the Toffoli gate  in $\TOF$ behaves as expected on these points; the proof is contained in Appendix \ref{appendix:BehavesAsExpected}:

\begin{lemma}
	\label{lemma:BehavesAsExpected}
	
	$\tof \circ |x_1,x_2,x_3\> = | x_1,x_2,x_1\cdot x_2 + x_3\>$
\end{lemma}

The points above which are expressed in bra notation are clearly total.  However, not all maps with domain $0$ are total: in particular, as in $\CNOT$, the maps $\Omega_{0,m}$ for $m \in \N$ are not total:

\begin{definition}
	\label{definition:Omega}
	Define $\Omega:=  \zeroout\circ\onein$, and $\Omega_{n,m}:=\onein^{\otimes^m}\circ\Omega\circ  \oneout^{\otimes^n}$.
\end{definition}

When any map is tensored with such a map it becomes one of them:

\begin{lemma}
\label{lemma:DegenerateSame}
For all circuits $f\in\TOF(n,m)$, $f \otimes \Omega = \Omega_{n,m}$.
\end{lemma}
\begin{proof}
By the functionality of the interpretation of $\CNOT$ into $\TOF$ and \cite[Lemma A.2]{CNOT} note that $\Omega$ is idempotent.  Therefore, use \ref{TOF.7} to cut all of the wires around all Toffoli gates.  Notice that when the wires round a Toffoli gate is cut this results in the map:
$$  \<111|\circ\tof\circ|111\>  =  \<11|\circ\cnot\circ|11\>   =  \<1|\circ \Not\circ |1\>  = \<1|\circ |0\> =: \Omega$$
When all wires of a circuit are cut it will therefore take the form:
$$\onein^{\otimes^m} \circ (\oneout\circ \onein)^{\otimes^\ell} \circ \Omega^{\otimes^k}\circ  \oneout^{\otimes^n}
= \onein^{\otimes^m} \circ 1_0^{\otimes^\ell} \circ  \Omega^{\otimes^k} \circ \oneout^{\otimes^n}
 =\onein^{\otimes^m} \circ \Omega \circ \oneout^{\otimes^n}
 =:\Omega_{n,m}
$$
\end{proof}

In fact, there are only two sorts of points, those that are expressible as  $| x_1,\cdots, x_n\>$---and maps of the form $\Omega_{0,n}$ for some $n \in \N$:

\begin{lemma}
\label{lemma:totalOrDegenerate}
	For all $n \in \N$ and $f\in\TOF(0,n)$ implies $f=\Omega_{0,n}$ or $f=|b_1,\cdots, b_n\>$ for some $b_1,\cdots,b_n \in \Z_2$.
\end{lemma}
\begin{proof}
Given any circuit $f \in \TOF(0,n)$, pull all of the $\onein$ gates to the left of the circuit and all of the  $\oneout$ gate to the right.  In the middle, there will only be $\tof$ gates; thus, apply the $\onein$ gates to the $\tof$ gates using Lemma \ref{lemma:BehavesAsExpected}.  This will result in a series of $\onein$ and $\zeroin$ gates on the left.  Repeatedly apply these $\onein$ and $\zeroin$ gates to the $\tof$ gates using Lemma \ref{lemma:BehavesAsExpected} until there are only a series of $\onein$ and $\zeroin$ gates on the left and a series of $\oneout$ gates on the right, and nothing in the middle.  If a $\onein$ and $\oneout$ gate meet, they compose to form the identity and disappear.  If this process can eliminate all of the $\oneout$ gates, then we are done.  Otherwise, if  $\zeroin$ and $\oneout$ gate meet, they compose to form $\Omega$; therefore, by Lemma \ref{lemma:DegenerateSame}, $f = \Omega_{0,n}$.
\end{proof}


\section{Partial injective functions and \texorpdfstring{$\TOF$}{TOF}}
\label{section:thefunctor}


For any restriction category $\X$ and any object $X \in \X$, we may define a restriction-preserving hom-functor $h_X:\X\to\Par$ by:
\begin{description}
	\item[On Objects:] $h_X(Y):=\{f\in\X(X,Y)| \bar f = 1_X\}$
	\item[On Maps:] For each map $f: Y \to Z$, for all $g \in h_X(Y)$:
		\[
		(h_X(f))(g)
		:=
		\begin{cases}
		gf & \text{ if } \bar{gf} = 1_X\\
		\uparrow & \text{ otherwise}
		\end{cases}
		\]	
\end{description}
As $\TOF$ is a restriction category, we may consider the functor  $h_0: \TOF \to \Par$ which evaluates circuits on the total points of $\TOF$.  As $\TOF$ is an inverse category and this functor preserves restriction, we can lift $h_0:\TOF\to \Par$ to a restriction-preserving functor $H_0:\TOF\to\Pinj$.  Clearly the nature of the total points of $\TOF$ ensures that the discrete inverse structure is preserved by $H_0$:

\begin{lemma}
$H_0:\TOF\to\Pinj$ preserves discrete products.
\end{lemma}

\begin{proof}
$H_0(n) = \{ |b_1...b_n\> | b_i \in \{ 0,1\} \}$ so that $H_0(n+m) \simeq H_0(n) \times H_0(m)$ and $H_0(\Delta) = \Delta$.
\end{proof}

Define $\FPinj_2$ to be the full subcategory of $\Pinj$ (sets and partial isomorphisms) with objects finite powers of the two element set.   By Lemma \ref{lemma:totalOrDegenerate}, the object $n \in \TOF$ under $H_0$ corresponds to the set $\{ |0\>,|1\>\}^n$ in  $\FPinj_2$; therefore, by restricting the co-domain of $H_0$ we obtain a bijective on objects functor $\tilde H_0:\TOF\to\FPinj_2$.  We will prove that this functor is an equivalence of categories.  


\section{A normal form for the idempotents of \texorpdfstring{$\TOF$}{TOF}}
\label{section:NormalForm}


As in \cite{CNOT}, our objective is to reduce the fullness and faithfulness of $\tilde H_0:\TOF\to\FPinj_2$ to its fullness and faithfulness on restriction  idempotents. $\cnot_n$ gates will be used to define a class of circuits which will allow a normal form for the idempotents of $\TOF$ to be established.  It is also worth noting that the following class of circuits corresponds to the canonical form of \cite[Def. 4]{Iwama}:

\begin{definition}
	\label{definition:polynomialForm}
	A circuit $f:n\to n$ is said to be in \textbf{polynomial form} when it is the composition of circuits $f=c_1\cdots c_k$ where each $c_i$ is a $\cnot_j$ gate with control any of the first $n-1$ wires and target the last wire.  
\end{definition}	

For example, the following circuit corresponding to the polynomial $x_1x_2+x_2x_4+x_1x_2+x_1x_2+x_2x_3x_4+x_4$ is in polynomial form:
	$$
	
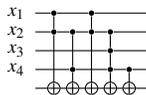
\begin{tikzpicture}
	\begin{pgfonlayer}{nodelayer}
		\node [style=dot] (0) at (0.5000001, 2) {};
		\node [style=dot] (1) at (0.5000001, 1.5) {};
		\node [style=dot] (2) at (1, 1.5) {};
		\node [style=dot] (3) at (1, 0.5000001) {};
		\node [style=dot] (4) at (1.5, 2) {};
		\node [style=dot] (5) at (1.5, 1.5) {};
		\node [style=dot] (6) at (2, 1.5) {};
		\node [style=dot] (7) at (2, 1) {};
		\node [style=dot] (8) at (2, 0.5000001) {};
		\node [style=dot] (9) at (2.5, 0.5000001) {};
		\node [style=oplus] (10) at (0.5000001, -0) {};
		\node [style=oplus] (11) at (1, -0) {};
		\node [style=oplus] (12) at (1.5, -0) {};
		\node [style=oplus] (13) at (2, -0) {};
		\node [style=oplus] (14) at (2.5, -0) {};
		\node [style=nothing] (15) at (3, 0) {};
		\node [style=nothing] (16) at (3, 1) {};
		\node [style=nothing] (17) at (3, 2) {};
		\node [style=nothing] (18) at (3, 1.5) {};
		\node [style=nothing] (19) at (3, 0.5000001) {};
		\node [style=nothing] (20) at (0, -0) {};
		\node [style=nothing] (21) at (0, 1) {};
		\node [style=nothing] (22) at (0, 2) {};
		\node [style=nothing] (23) at (0, 1.5) {};
		\node [style=nothing] (24) at (0, 0.5000001) {};
		\node [style=nothing] (25) at (-0.5, 2) {\Large $x_1$};
		\node [style=nothing] (26) at (-0.5, 1.5) {\Large $x_2$};
		\node [style=nothing] (27) at (-0.5, 1) {\Large $x_3$};
		\node [style=nothing] (28) at (-0.5, 0.5) {\Large $x_4$};
	\end{pgfonlayer}
	\begin{pgfonlayer}{edgelayer}
		\draw (22) to (0);
		\draw (0) to (4);
		\draw (4) to (17);
		\draw (18) to (6);
		\draw (6) to (5);
		\draw (5) to (2);
		\draw (2) to (1);
		\draw (1) to (23);
		\draw (21) to (7);
		\draw (7) to (16);
		\draw (19) to (9);
		\draw (9) to (8);
		\draw (8) to (3);
		\draw (3) to (24);
		\draw (20) to (10);
		\draw (10) to (11);
		\draw (11) to (12);
		\draw (12) to (13);
		\draw (13) to (14);
		\draw (14) to (15);
		\draw (14) to (9);
		\draw (13) to (8);
		\draw (8) to (7);
		\draw (7) to (6);
		\draw (5) to (4);
		\draw (5) to (12);
		\draw (11) to (3);
		\draw (3) to (2);
		\draw (0) to (1);
		\draw (1) to (10);
	\end{pgfonlayer}
\end{tikzpicture}
	$$

Clearly all polynomials (in normal form) can be represented by a circuit in polynomial form, where each monomial corresponds to a $\cnot_n$ gate.  Moreover, this correspondence is bijective, as $\cnot_n$ gates targeting the same wire obviously commute, and $\cnot_n$ gates are self-inverse by Lemma \ref{lemma:PolySelfInverse}. 		
	
	For example, the following identity holds in $\Z_2[x_1,x_2,x_3,x_4]$
	$$ x_1x_2+x_2x_4+x_1x_2+x_2x_3x_4+x_4 =  x_2 x_4+x_2x_3x_4+x_4$$
	and this corresponds to the following identity in $\TOF$:
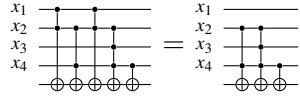
	$$
	\begin{tikzpicture}
	\begin{pgfonlayer}{nodelayer}
	\node [style=dot] (0) at (0.5000001, 2) {};
	\node [style=dot] (1) at (0.5000001, 1.5) {};
	\node [style=dot] (2) at (1, 1.5) {};
	\node [style=dot] (3) at (1, 0.5000001) {};
	\node [style=dot] (4) at (1.5, 2) {};
	\node [style=dot] (5) at (1.5, 1.5) {};
	\node [style=dot] (6) at (2, 1.5) {};
	\node [style=dot] (7) at (2, 1) {};
	\node [style=dot] (8) at (2, 0.5000001) {};
	\node [style=dot] (9) at (2.5, 0.5000001) {};
	\node [style=oplus] (10) at (0.5000001, -0) {};
	\node [style=oplus] (11) at (1, -0) {};
	\node [style=oplus] (12) at (1.5, -0) {};
	\node [style=oplus] (13) at (2, -0) {};
	\node [style=oplus] (14) at (2.5, -0) {};
	\node [style=nothing] (15) at (3, 0) {};
	\node [style=nothing] (16) at (3, 1) {};
	\node [style=nothing] (17) at (3, 2) {};
	\node [style=nothing] (18) at (3, 1.5) {};
	\node [style=nothing] (19) at (3, 0.5000001) {};
	\node [style=nothing] (20) at (0, -0) {};
	\node [style=nothing] (21) at (0, 1) {};
	\node [style=nothing] (22) at (0, 2) {};
	\node [style=nothing] (23) at (0, 1.5) {};
	\node [style=nothing] (24) at (0, 0.5000001) {};
	\node [style=nothing] (25) at (-0.5, 2) {\Large $x_1$};
	\node [style=nothing] (26) at (-0.5, 1.5) {\Large $x_2$};
	\node [style=nothing] (27) at (-0.5, 1) {\Large $x_3$};
	\node [style=nothing] (28) at (-0.5, 0.5) {\Large $x_4$};
	\end{pgfonlayer}
	\begin{pgfonlayer}{edgelayer}
	\draw (22) to (0);
	\draw (0) to (4);
	\draw (4) to (17);
	\draw (18) to (6);
	\draw (6) to (5);
	\draw (5) to (2);
	\draw (2) to (1);
	\draw (1) to (23);
	\draw (21) to (7);
	\draw (7) to (16);
	\draw (19) to (9);
	\draw (9) to (8);
	\draw (8) to (3);
	\draw (3) to (24);
	\draw (20) to (10);
	\draw (10) to (11);
	\draw (11) to (12);
	\draw (12) to (13);
	\draw (13) to (14);
	\draw (14) to (15);
	\draw (14) to (9);
	\draw (13) to (8);
	\draw (8) to (7);
	\draw (7) to (6);
	\draw (5) to (4);
	\draw (5) to (12);
	\draw (11) to (3);
	\draw (3) to (2);
	\draw (0) to (1);
	\draw (1) to (10);
	\end{pgfonlayer}
	\end{tikzpicture}
	=
	\begin{tikzpicture}
	\begin{pgfonlayer}{nodelayer}
	\node [style=dot] (0) at (0.5, 1.5) {};
	\node [style=dot] (1) at (0.5, 0.5) {};
	\node [style=dot] (2) at (1, 1.5) {};
	\node [style=dot] (3) at (1, 1) {};
	\node [style=dot] (4) at (1, 0.5) {};
	\node [style=dot] (5) at (1.5, 0.5) {};
	\node [style=oplus] (6) at (0.5, 0) {};
	\node [style=oplus] (7) at (1, 0) {};
	\node [style=oplus] (8) at (1.5, 0) {};
	\node [style=nothing] (9) at (2, 0) {};
	\node [style=nothing] (10) at (2, 1) {};
	\node [style=nothing] (11) at (2, 2) {};
	\node [style=nothing] (12) at (2, 1.5) {};
	\node [style=nothing] (13) at (2, 0.5) {};
	\node [style=nothing] (14) at (0, 0) {};
	\node [style=nothing] (15) at (0, 1) {};
	\node [style=nothing] (16) at (0, 2) {};
	\node [style=nothing] (17) at (0, 1.5) {};
	\node [style=nothing] (18) at (0, 0.5) {};
	\node [style=nothing] (19) at (-0.5, 2) {\Large $x_1$};
	\node [style=nothing] (20) at (-0.5, 1.5) {\Large $x_2$};
	\node [style=nothing] (21) at (-0.5, 1) {\Large $x_3$};
	\node [style=nothing] (22) at (-0.5, 0.5) {\Large $x_4$};
	\end{pgfonlayer}
	\begin{pgfonlayer}{edgelayer}
	\draw (12) to (2);
	\draw (15) to (3);
	\draw (3) to (10);
	\draw (13) to (5);
	\draw (5) to (4);
	\draw (4) to (1);
	\draw (1) to (18);
	\draw (7) to (8);
	\draw (8) to (9);
	\draw (8) to (5);
	\draw (7) to (4);
	\draw (4) to (3);
	\draw (3) to (2);
	\draw (6) to (1);
	\draw (1) to (0);
	\draw [style=simple] (16) to (11);
	\draw [style=simple] (2) to (0);
	\draw [style=simple] (0) to (17);
	\draw [style=simple] (7) to (6);
	\draw [style=simple] (6) to (14);
	\end{pgfonlayer}
	\end{tikzpicture}
	$$

Since the restriction idempotents in $\FPinj$ are determined by boolean equations or, equivalently, by the zeros of multivariate polynomials over $\Z_2$, we  use polynomial  form circuits to provide a normal form for the restriction idempotents of $\TOF$.   By restricting the value of the bottom wire to be 0,  the evaluation of this circuit on total points is defined if and only if the corresponding tuple of bits is in the kernel of the function corresponding to polynomial evaluation:

\begin{definition}
\label{definition:idempotentPolynomialForm}
A circuit $e:n\to n$ in $\TOF$ is a \textbf{polyform} if $e=( 1_n\ox\zeroout) \circ q \circ  ( 1_n\ox\zeroin)$ for some $q:n+1\to n+1$ in polynomial form.
\end{definition}

For example the following circuit corresponding to the equation $x_2x_4+x_2x_3x_4+x_4=0$ is a polyform:
$$
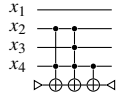
\begin{tikzpicture}
\begin{pgfonlayer}{nodelayer}
\node [style=dot] (0) at (0.5, 1.5) {};
\node [style=dot] (1) at (0.5, 0.5) {};
\node [style=dot] (2) at (1, 1.5) {};
\node [style=dot] (3) at (1, 1) {};
\node [style=dot] (4) at (1, 0.5) {};
\node [style=dot] (5) at (1.5, 0.5) {};
\node [style=oplus] (6) at (0.5, 0) {};
\node [style=oplus] (7) at (1, 0) {};
\node [style=oplus] (8) at (1.5, 0) {};
\node [style=nothing] (9) at (2, 1) {};
\node [style=nothing] (10) at (2, 2) {};
\node [style=nothing] (11) at (2, 1.5) {};
\node [style=nothing] (12) at (2, 0.5) {};
\node [style=nothing] (13) at (0, 1) {};
\node [style=nothing] (14) at (0, 2) {};
\node [style=nothing] (15) at (0, 1.5) {};
\node [style=nothing] (16) at (0, 0.5) {};
\node [style=nothing] (17) at (-0.5, 2) {\Large $x_1$};
\node [style=nothing] (18) at (-0.5, 1.5) {\Large $x_2$};
\node [style=nothing] (19) at (-0.5, 1) {\Large $x_3$};
\node [style=nothing] (20) at (-0.5, 0.5) {\Large $x_4$};
\node [style=zeroin] (21) at (0, 0) {};
\node [style=zeroout] (22) at (2, 0) {};
\end{pgfonlayer}
\begin{pgfonlayer}{edgelayer}
\draw (11) to (2);
\draw (13) to (3);
\draw (3) to (9);
\draw (12) to (5);
\draw (5) to (4);
\draw (4) to (1);
\draw (1) to (16);
\draw (7) to (8);
\draw (8) to (5);
\draw (7) to (4);
\draw (4) to (3);
\draw (3) to (2);
\draw (6) to (1);
\draw (1) to (0);
\draw [style=simple] (14) to (10);
\draw [style=simple] (2) to (0);
\draw [style=simple] (0) to (15);
\draw [style=simple] (7) to (6);
\draw [style=simple] (21) to (6);
\draw [style=simple] (22) to (8);
\end{pgfonlayer}
\end{tikzpicture}
$$

Given a polynomial form circuit $q$, let the circuit 
$\begin{tikzpicture}
\begin{pgfonlayer}{nodelayer}
\node [style=map] (0) at (0, -0) {$q$};
\node [style=nothing] (1) at (-0.5, -0) {};
\node [style=nothing] (2) at (0.5, -0) {};
\node [style=nothing] (3) at (-0.5, 0.5) {};
\node [style=nothing] (4) at (0.5, 0.5) {};
\node [style=dot] (5) at (0, 0.5) {};
\end{pgfonlayer}
\begin{pgfonlayer}{edgelayer}
\draw (3) to (5);
\draw (5) to (4);
\draw (2) to (0);
\draw (0) to (1);
\draw (0) to (5);
\end{pgfonlayer}
\end{tikzpicture}$
denote the circuit where every $\cnot_n$ gate of $q$ is additionally controlled by wire with the black dot.

In fact given a polynomial form circuit $q:n+1\to n+1$, by the repeated application of Lemma \ref{remark:Iwama} {\em (v)} for each constituent $\cnot_n$ gate, it is useful to observe:

$$
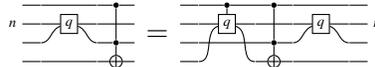
\begin{tikzpicture}
\begin{pgfonlayer}{nodelayer}
\node [style=map] (0) at (0, -0) {$q$};
\node [style=nothing] (1) at (-0.7500002, -0) {};
\node [style=nothing] (2) at (0.7500002, -0) {};
\node [style=nothing] (3) at (0.7500002, -0.5) {};
\node [style=nothing] (4) at (-0.7500002, -0.5) {};
\node [style=nothing] (5) at (-0.7500002, 0.5) {};
\node [style=nothing] (6) at (0.7500002, 0.5) {};
\node [style=nothing] (7) at (-0.7500002, -1) {};
\node [style=nothing] (8) at (0.7500002, -1) {};
\node [style=dot] (9) at (1.25, 0.5) {};
\node [style=dot] (10) at (1.25, -0.5) {};
\node [style=oplus] (11) at (1.25, -1) {};
\node [style=nothing] (12) at (1.75, 0.5) {};
\node [style=nothing] (13) at (1.75, 0) {};
\node [style=nothing] (14) at (1.75, -1) {};
\node [style=nothing] (15) at (1.75, -0.5) {};
\node [style=nothing] (16) at (-1.25, 0.5) {};
\node [style=nothing] (17) at (-1.25, -0) {};
\node [style=nothing] (18) at (-1.25, -1) {};
\node [style=nothing] (19) at (-1.25, -0.5) {};
\node [style=nothing] (20) at (-1.5, -0) {$n$};
\end{pgfonlayer}
\begin{pgfonlayer}{edgelayer}
\draw (5) to (6);
\draw (1) to (0);
\draw (0) to (2);
\draw [in=-15, out=180, looseness=1.00] (3) to (0);
\draw [in=0, out=-165, looseness=1.00] (0) to (4);
\draw (7) to (8);
\draw (16) to (5);
\draw (17) to (1);
\draw (4) to (19);
\draw (18) to (7);
\draw (6) to (9);
\draw (9) to (12);
\draw (13) to (2);
\draw (3) to (10);
\draw (10) to (15);
\draw (14) to (11);
\draw (11) to (8);
\draw (11) to (10);
\draw (10) to (9);
\end{pgfonlayer}
\end{tikzpicture}
=
\begin{tikzpicture}
\begin{pgfonlayer}{nodelayer}
\node [style=map] (0) at (0, -0) {$q$};
\node [style=nothing] (1) at (-0.7500002, -0) {};
\node [style=nothing] (2) at (0.7500002, -0) {};
\node [style=nothing] (3) at (0.7500002, -0.5) {};
\node [style=nothing] (4) at (-0.7500002, -0.5) {};
\node [style=nothing] (5) at (-0.7500002, 0.5) {};
\node [style=nothing] (6) at (0.7500002, 0.5) {};
\node [style=nothing] (7) at (-0.7500002, -1) {};
\node [style=nothing] (8) at (0.7500002, -1) {};
\node [style=nothing] (9) at (1.25, 0.5) {};
\node [style=nothing] (10) at (1.25, -0) {};
\node [style=nothing] (11) at (1.25, -1) {};
\node [style=nothing] (12) at (1.25, -0.5) {};
\node [style=dot] (13) at (-1.25, -0.5) {};
\node [style=oplus] (14) at (-1.25, -1) {};
\node [style=dot] (15) at (-1.25, 0.5) {};
\node [style=nothing] (16) at (-3.25, -1) {};
\node [style=nothing] (17) at (-3.25, -0) {};
\node [style=map] (18) at (-2.5, -0) {$q$};
\node [style=nothing] (19) at (-1.75, -0) {};
\node [style=nothing] (20) at (-1.75, -1) {};
\node [style=nothing] (21) at (-3.25, -0.5) {};
\node [style=nothing] (22) at (-3.25, 0.5) {};
\node [style=dot] (23) at (-2.5, 0.5) {};
\node [style=nothing] (24) at (-3.75, -1) {};
\node [style=nothing] (25) at (-3.75, -0) {};
\node [style=nothing] (26) at (-3.75, -0.5) {};
\node [style=nothing] (27) at (-3.75, 0.5) {};
\node [style=nothing] (28) at (1.5, -0) {$n$};
\end{pgfonlayer}
\begin{pgfonlayer}{edgelayer}
\draw (5) to (6);
\draw (1) to (0);
\draw (0) to (2);
\draw [in=-15, out=180, looseness=1.00] (3) to (0);
\draw [in=0, out=-165, looseness=1.00] (0) to (4);
\draw (7) to (8);
\draw (10) to (2);
\draw (14) to (13);
\draw (13) to (15);
\draw (6) to (9);
\draw (12) to (3);
\draw (8) to (11);
\draw (15) to (5);
\draw (4) to (13);
\draw (7) to (14);
\draw (17) to (18);
\draw (18) to (19);
\draw [in=-15, out=180, looseness=1.00] (20) to (18);
\draw [in=0, out=-165, looseness=1.00] (18) to (16);
\draw (20) to (14);
\draw (19) to (1);
\draw (18) to (23);
\draw (23) to (15);
\draw (23) to (22);
\draw (21) to (13);
\draw (27) to (22);
\draw (25) to (17);
\draw (21) to (26);
\draw (24) to (16);
\end{pgfonlayer}
\end{tikzpicture}
$$

Our aim is now to show that the idempotents of $\TOF$ are always equivalent to a polyform so that these circuits provide a normal form for the idempotents.   Because some of the steps are subtle, we annotate certain equalities with the corresponding identities used.

\begin{lemma}
\label{lemma:IdempotentPolynomialsAreIdempotent}
Polyforms are idempotents (and therefore restriction idempotents).
\end{lemma}

\begin{proof}
Consider some map $e:=( 1_n\ox\zeroout) \circ q \circ  ( 1_n\ox\zeroin)$ a polyform, as above, then:
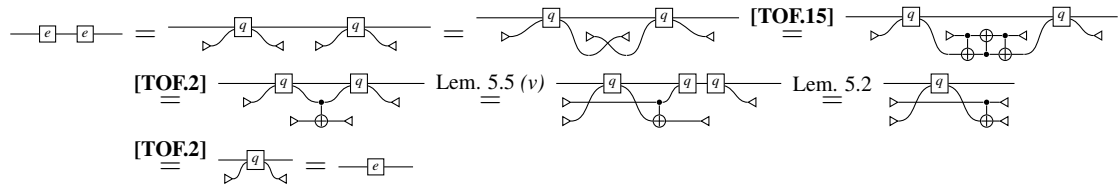
\begin{align*}
\begin{tikzpicture}
	\begin{pgfonlayer}{nodelayer}
		\node [style=nothing] (0) at (0, -0) {};
		\node [style=map] (1) at (1, -0) {$e$};
		\node [style=map] (2) at (2, -0) {$e$};
		\node [style=nothing] (3) at (3, -0) {};
	\end{pgfonlayer}
	\begin{pgfonlayer}{edgelayer}
		\draw [style=simple] (0) to (1);
		\draw (1) to (2);
		\draw (2) to (3);
	\end{pgfonlayer}
\end{tikzpicture}
&=
\begin{tikzpicture}
	\begin{pgfonlayer}{nodelayer}
		\node [style=map] (0) at (2, -0) {$q$};
		\node [style=map] (1) at (5, -0) {$q$};
		\node [style=nothing] (2) at (0, -0) {};
		\node [style=nothing] (3) at (7, -0) {};
		\node [style=zeroin] (4) at (1, -0.5) {};
		\node [style=zeroout] (5) at (3, -0.5) {};
		\node [style=zeroin] (6) at (4, -0.5) {};
		\node [style=zeroout] (7) at (6, -0.5) {};
	\end{pgfonlayer}
	\begin{pgfonlayer}{edgelayer}
		\draw (2) to (0);
		\draw (0) to (1);
		\draw (1) to (3);
		\draw [in=-165, out=0, looseness=1.00] (4) to (0);
		\draw [in=180, out=-15, looseness=1.00] (0) to (5);
		\draw [in=-165, out=0, looseness=1.00] (6) to (1);
		\draw [in=-15, out=180, looseness=1.00] (7) to (1);
	\end{pgfonlayer}
\end{tikzpicture}
=
\begin{tikzpicture}
	\begin{pgfonlayer}{nodelayer}
		\node [style=map] (0) at (2, -0) {$q$};
		\node [style=map] (1) at (5, -0) {$q$};
		\node [style=nothing] (2) at (0, -0) {};
		\node [style=nothing] (3) at (7, -0) {};
		\node [style=zeroin] (4) at (0.75, -0.5) {};
		\node [style=zeroout] (5) at (6.25, -0.5) {};
		\node [style=zeroin] (6) at (3, -0.5000001) {};
		\node [style=zeroout] (7) at (4, -0.5000001) {};
		\node [style=nothing] (8) at (3, -0.9999999) {};
		\node [style=nothing] (9) at (4, -0.9999999) {};
	\end{pgfonlayer}
	\begin{pgfonlayer}{edgelayer}
		\draw (2) to (0);
		\draw (0) to (1);
		\draw (1) to (3);
		\draw [in=-165, out=0, looseness=1.00] (4) to (0);
		\draw [in=-15, out=180, looseness=1.00] (5) to (1);
		\draw [in=180, out=-15, looseness=1.00] (0) to (8);
		\draw [in=180, out=0, looseness=1.00] (8) to (7);
		\draw [in=0, out=-165, looseness=1.00] (1) to (9);
		\draw [in=0, out=180, looseness=1.00] (9) to (6);
	\end{pgfonlayer}
\end{tikzpicture}
\myeq{\ref{TOF.14}}
\begin{tikzpicture}
	\begin{pgfonlayer}{nodelayer}
		\node [style=map] (0) at (2, -0) {$q$};
		\node [style=map] (1) at (6, -0) {$q$};
		\node [style=nothing] (2) at (0.25, -0) {};
		\node [style=nothing] (3) at (7.5, -0) {};
		\node [style=zeroin] (4) at (1, -0.5) {};
		\node [style=zeroout] (5) at (7, -0.5) {};
		\node [style=nothing] (6) at (3, -0.9999999) {};
		\node [style=nothing] (7) at (5, -0.9999999) {};
		\node [style=zeroout] (8) at (5, -0.5000001) {};
		\node [style=zeroin] (9) at (3, -0.5000001) {};
		\node [style=dot] (10) at (3.5, -0.5000001) {};
		\node [style=dot] (11) at (4.5, -0.5000001) {};
		\node [style=dot] (12) at (4, -0.9999999) {};
		\node [style=oplus] (13) at (4, -0.5000001) {};
		\node [style=oplus] (14) at (3.5, -0.9999999) {};
		\node [style=oplus] (15) at (4.5, -0.9999999) {};
	\end{pgfonlayer}
	\begin{pgfonlayer}{edgelayer}
		\draw (2) to (0);
		\draw (0) to (1);
		\draw (1) to (3);
		\draw [in=-165, out=0, looseness=1.00] (4) to (0);
		\draw [in=-15, out=180, looseness=1.00] (5) to (1);
		\draw [in=180, out=-15, looseness=1.00] (0) to (6);
		\draw [in=0, out=-165, looseness=1.00] (1) to (7);
		\draw (6) to (14);
		\draw (14) to (12);
		\draw (12) to (15);
		\draw (15) to (7);
		\draw (8) to (11);
		\draw (11) to (13);
		\draw (13) to (10);
		\draw (10) to (9);
		\draw (14) to (10);
		\draw (13) to (12);
		\draw (15) to (11);
	\end{pgfonlayer}
\end{tikzpicture}\\
&\myeq{\ref{TOF.2}}
\begin{tikzpicture}
	\begin{pgfonlayer}{nodelayer}
		\node [style=map] (0) at (1.75, -0) {$q$};
		\node [style=map] (1) at (3.75, -0) {$q$};
		\node [style=nothing] (2) at (0, -0) {};
		\node [style=nothing] (3) at (5.5, -0) {};
		\node [style=zeroin] (4) at (0.75, -0.5) {};
		\node [style=zeroout] (5) at (4.75, -0.5) {};
		\node [style=dot] (6) at (2.75, -0.5) {};
		\node [style=oplus] (7) at (2.75, -1) {};
		\node [style=zeroin] (8) at (2, -1) {};
		\node [style=zeroout] (9) at (3.5, -1) {};
	\end{pgfonlayer}
	\begin{pgfonlayer}{edgelayer}
		\draw (2) to (0);
		\draw (0) to (1);
		\draw (1) to (3);
		\draw [in=-165, out=0, looseness=1.00] (4) to (0);
		\draw [in=-15, out=180, looseness=1.00] (5) to (1);
		\draw [style=simple, in=180, out=-15, looseness=1.25] (0) to (6);
		\draw [style=simple, in=0, out=-165, looseness=1.25] (1) to (6);
		\draw [style=simple] (6) to (7);
		\draw [style=simple] (8) to (7);
		\draw [style=simple] (7) to (9);
	\end{pgfonlayer}
\end{tikzpicture}
\myeq{Lem. \ref{remark:Iwama} {\em (v)}}
\begin{tikzpicture}
	\begin{pgfonlayer}{nodelayer}
		\node [style=map] (0) at (4.25, -0) {$q$};
		\node [style=nothing] (1) at (0, -0) {};
		\node [style=nothing] (2) at (6, -0) {};
		\node [style=zeroin] (3) at (0.2500001, -0.5000001) {};
		\node [style=zeroout] (4) at (5.25, -0.5000001) {};
		\node [style=zeroout] (5) at (4, -1) {};
		\node [style=map] (6) at (3.5, -0) {$q$};
		\node [style=map] (7) at (1.5, -0) {$q$};
		\node [style=zeroin] (8) at (0.2500001, -1) {};
		\node [style=oplus] (9) at (2.75, -1) {};
		\node [style=dot] (10) at (2.75, -0.5000001) {};
	\end{pgfonlayer}
	\begin{pgfonlayer}{edgelayer}
		\draw (0) to (2);
		\draw [in=-15, out=180, looseness=1.00] (4) to (0);
		\draw [style=simple] (10) to (9);
		\draw (6) to (0);
		\draw [in=0, out=-165, looseness=1.00] (6) to (10);
		\draw (10) to (3);
		\draw [in=180, out=-15, looseness=1.00] (7) to (9);
		\draw (9) to (5);
		\draw [in=0, out=-165, looseness=1.00] (7) to (8);
		\draw (1) to (7);
		\draw (7) to (6);
	\end{pgfonlayer}
\end{tikzpicture}
\myeq{Lem. \ref{lemma:PolySelfInverse}}
\begin{tikzpicture}
	\begin{pgfonlayer}{nodelayer}
		\node [style=nothing] (0) at (0, -0) {};
		\node [style=nothing] (1) at (3.5, -0) {};
		\node [style=zeroin] (2) at (0.2500001, -0.5000001) {};
		\node [style=zeroout] (3) at (3.25, -0.5000001) {};
		\node [style=zeroout] (4) at (3.25, -0.9999999) {};
		\node [style=map] (5) at (1.5, -0) {$q$};
		\node [style=zeroin] (6) at (0.2500001, -1) {};
		\node [style=oplus] (7) at (2.75, -1) {};
		\node [style=dot] (8) at (2.75, -0.5000001) {};
	\end{pgfonlayer}
	\begin{pgfonlayer}{edgelayer}
		\draw [style=simple] (8) to (7);
		\draw (8) to (2);
		\draw [in=180, out=-15, looseness=1.00] (5) to (7);
		\draw (7) to (4);
		\draw [in=0, out=-165, looseness=1.00] (5) to (6);
		\draw (0) to (5);
		\draw (5) to (1);
		\draw (3) to (8);
	\end{pgfonlayer}
\end{tikzpicture}\\
&\myeq{\ref{TOF.2}}
\begin{tikzpicture}
	\begin{pgfonlayer}{nodelayer}
		\node [style=nothing] (0) at (0.5000001, -0) {};
		\node [style=nothing] (1) at (2.5, -0) {};
		\node [style=zeroout] (2) at (2.25, -0.5000001) {};
		\node [style=map] (3) at (1.5, -0) {$q$};
		\node [style=zeroin] (4) at (0.7500001, -0.5000001) {};
	\end{pgfonlayer}
	\begin{pgfonlayer}{edgelayer}
		\draw [in=0, out=-165, looseness=1.00] (3) to (4);
		\draw (0) to (3);
		\draw (3) to (1);
		\draw [in=180, out=-15, looseness=1.00] (3) to (2);
	\end{pgfonlayer}
\end{tikzpicture}
=
\begin{tikzpicture}
	\begin{pgfonlayer}{nodelayer}
		\node [style=nothing] (0) at (0, -0) {};
		\node [style=map] (1) at (1, -0) {$e$};
		\node [style=nothing] (2) at (2, -0) {};
	\end{pgfonlayer}
	\begin{pgfonlayer}{edgelayer}
		\draw [style=simple] (0) to (1);
		\draw (1) to (2);
	\end{pgfonlayer}
\end{tikzpicture}
\end{align*}

\end{proof}

We prove that the converse also holds using 3 lemmas:
\begin{proposition}
\label{lemma:IdempotentsAreIdempotentPolynomials}
All idempotents in $\TOF$ are equivalent to a circuit which is a polyform.
\end{proposition}
Given a polyform, $f$, by structural induction, it suffices to show that sandwiching $f$ with a generating gate, $g$, to obtain $(1 \otimes g \otimes 1) f( 1 \otimes g^\cnv \otimes 1)$, always results in a circuit which has a polyform.  There are 3 cases, one for each generating circuit:

\begin{lemma}
\label{lemma:lemma:IdempotentsAreIdempotentPolynomialsTof}
If $f:n\to n$ is a polyform then sandwiching $f$ by a $\tof$ gate results in a circuit with a polyform.
\end{lemma}

\begin{proof}
By Lemma \ref{remark:Iwama} {\em (v)}, push $g$ through $f$ until it meets the other $g^\cnv= g$ and then annihilate both generalized controlled not gates by Lemma \ref{lemma:PolySelfInverse}.  This circuit is still a polyform.
\end{proof}

\begin{lemma}
\label{lemma:lemma:IdempotentsAreIdempotentPolynomialsOnein} 
If $f:n\to n$ is a polyform then sandwiching $f$ by a $\onein$ gate results in a circuit equivalent to a polyform.
\end{lemma}

\begin{proof}
It suffices to observe for the inductive step that:
		\begin{align*}

		\end{align*}
\end{proof}

\begin{lemma}
\label{lemma:lemma:IdempotentsAreIdempotentPolynomialsOneout}
If $f:n\to n$ is a polyform then sandwiching $f$ by a $\onein$ and $\oneout$ gate results in a circuit with a polyform.
\end{lemma}

\begin{proof}
Suppose $f:n\to n$ is a polyform $(1_{n}\otimes \zeroout\otimes 1_{m}) \circ h \circ (1_{n}\otimes \zeroin \otimes 1_{m})$, then:

	\begin{align*}
%
		\end{align*}
\end{proof}

We have now established with Lemma \ref{lemma:IdempotentPolynomialsAreIdempotent} and Proposition \ref{lemma:IdempotentsAreIdempotentPolynomials}:

\begin{proposition}
\label{proposition:IdempotentNormalForm}
Polyforms are a normal form for the idempotents in $\TOF$. 
\end{proposition}

This implies the desired result:

\begin{corollary}
\label{proposition:H0FFRestIdems}
\label{proposition:IdempotentEquivalence}
$\tilde H_0:\TOF\to\FPinj_2$ is full and faithful on restriction idempotents.
\end{corollary}

\section{Full and faithfullness of \texorpdfstring{$\tilde H_0:\TOF\to\FPinj_2$}{the functor from TOF to FPInj2}}
\label{section:Equivalence}

We follow the approach of \cite{CNOT} to prove the fullness of $\tilde H_0:\TOF\to\FPinj_2$.  First we show that we can ``simulate'' all the total maps in $\FPinj_2$ with extra outputs:

\begin{lemma}
\label{lemma:ImageMap}
For every total map $f\in\FPinj_2(n,m)$, there is some $g\in\TOF$ such that $\tilde H_0(g)= \< 1_{\Z_2^n},  f \>$.
\end{lemma}
\begin{proof}
Consider a total $f \in \FPinj_2(n,m)$.  For any $i$ such that $1\leq i \leq m$, observe that $f\pi_i$ corresponds to a polynomial in $\Z[x_1,\cdots, x_n]$ and thus, there is a circuit $g_i:n+1\to n+1$ in polynomial form such that $\tilde H_0(h_i) = \< 1_{\Z_2^n},  f\pi_i\>$ where $h_i:=g_i \circ (1_n\otimes |0\>)$.

Now, inductively define the circuit $P_i$ for all $i$ such that $0 \leq i \leq m$, such that: 	$P_0=1_n$ and for every $i$ such that $0\leq i < m$, $P_{i+1}:=(P_{i}\ox 1_1)\circ  h_{i+1}$.

Then $\tilde H_0(P_m) = \< 1_{\Z_2^n} , f\pi_1 , \cdots ,  f\pi_m \> = \< 1_n, f\> $.
\end{proof}

Then we recall a technical observation from \cite{CNOT}:

\begin{lemma}[ {\cite[Lemma C.21]{CNOT}} ]
\label{lemma:fullCopy}
Let $F:\X \to \Y$ be an inverse product preserving functor between discrete inverse categories. Let $f$ be a partial isomorphism in $\mathbb{Y}$. If $\la \bar{f}, f \ra := \Delta(\bar f \otimes {f})$ and $\la \bar{f^\cnv}, f^\cnv \ra:=  \Delta(\bar{f^\cnv} \otimes {f^\cnv})$ are in the image of $F$, then $f$ and $f^\cnv$ are also in the image of $F$.
\end{lemma}

\oneraggedpage

This allows us to show:

\begin{proposition}
\label{proposition:H0Full}
$\tilde H_0:\TOF\to\FPinj_2$ is full.
\end{proposition}
\begin{proof}
Consider a map $f \in \FPinj_2(\Z_2^n,\Z_2^m)$ for arbitrary $n,m\in \Z$.  By Lemma \ref{lemma:fullCopy}, if we can simulate $\< \bar f,f\>$ and $\< \bar{f^\cnv}, f^\cnv \>$, we can simulate $f$. However, partial maps of the form $\< \bar g,g \>: X \to X \times Y$ are restrictions of a total map, unless $Y$ is empty,.   The case of $Y$ being empty does not occur in $\FPinj_2$ as the empty set is not an object.  Thus, all such maps are restrictions of total maps.  Therefore, by Lemma  \ref{lemma:ImageMap} there is some $h \in \TOF$ such that $\tilde H_0(h) \geq \< 1, g \>$ for any $g$.  However, by Proposition \ref{proposition:IdempotentEquivalence} $\tilde H_0:\TOF\to\FPinj_2$ is full on restriction idempotents, so there is some $e = \bar e$ such that $H_0(e)=\bar g$ and so $H_0(e h) = \<\bar g,g\>$ which  completes the proof.
\end{proof}

The faithfulness of $\tilde H_0$ is reduced to its faithfulness on restriction idempotents.  This uses another technical result from \cite{CNOT} which we recall:

\begin{lemma}[ {\cite[Lemma C.25]{CNOT}} ]
\label{lemma:faithfullemma}
A restriction functor $F:\mathbb{X}\to \mathbb{Y}$ between discrete inverse categories, which preserves inverse products, is faithful if and only if it is faithful on restriction idempotents.
\end{lemma}

As $\tilde H_0$ is faithful on idempotents and preserves inverse products, this gives:

\begin{proposition}
	\label{proposition:H0Faithful}
$\tilde H_0:\TOF\to\FPinj_2$ is faithful.
\end{proposition}

By Propositions  \ref{proposition:H0Full} and \ref{proposition:H0Faithful}, we may conclude:

\begin{theorem}
\label{theorem:H0Equiv}
$\TOF$ is discrete-inverse-equivalent to $\FPinj_2$.
\end{theorem}

\bibliographystyle{eptcs}

\bibliography{TOF}

\begin{thebibliography}{10}
\providecommand{\bibitemdeclare}[2]{}
\providecommand{\surnamestart}{}
\providecommand{\surnameend}{}
\providecommand{\urlprefix}{Available at }
\providecommand{\url}[1]{\texttt{#1}}
\providecommand{\href}[2]{\texttt{#2}}
\providecommand{\urlalt}[2]{\href{#1}{#2}}
\providecommand{\doi}[1]{doi:\urlalt{http://dx.doi.org/#1}{#1}}
\providecommand{\arxiv}[1]{\urlalt{http://arxiv.org/abs/#1}{#1}}
\providecommand{\bibinfo}[2]{#2}

\bibitemdeclare{inproceedings}{aaronson}
\bibitem{aaronson}
\bibinfo{author}{Scott \surnamestart Aaronson\surnameend},
  \bibinfo{author}{Daniel \surnamestart Grier\surnameend} \&
  \bibinfo{author}{Luke \surnamestart Schaeffer\surnameend}
  (\bibinfo{year}{2017}): \emph{\bibinfo{title}{{The Classification of
  Reversible Bit Operations}}}.
\newblock In \bibinfo{editor}{Christos~H. \surnamestart
  Papadimitriou\surnameend}, editor: {\sl \bibinfo{booktitle}{8th Innovations
  in Theoretical Computer Science Conference (ITCS 2017)}}, {\sl
  \bibinfo{series}{Leibniz International Proceedings in Informatics
  (LIPIcs)}}~\bibinfo{volume}{67}, \bibinfo{publisher}{Schloss
  Dagstuhl--Leibniz-Zentrum fuer Informatik}, \bibinfo{address}{Dagstuhl,
  Germany}, pp. \bibinfo{pages}{23:1--23:34},
  \doi{10.4230/LIPIcs.ITCS.2017.23}.

\bibitemdeclare{inproceedings}{inverse1}
\bibitem{inverse1}
\bibinfo{author}{Holger~Bock \surnamestart Axelsen\surnameend} \&
  \bibinfo{author}{Robin \surnamestart Kaarsgaard\surnameend}
  (\bibinfo{year}{2016}): \emph{\bibinfo{title}{Join inverse categories as
  models of reversible recursion}}.
\newblock In: {\sl \bibinfo{booktitle}{International Conference on Foundations
  of Software Science and Computation Structures}},
  \bibinfo{organization}{Springer}, pp. \bibinfo{pages}{73--90},
  \doi{10.1007/978-3-662-49630-5\_5}.

\bibitemdeclare{article}{Barr}
\bibitem{Barr}
\bibinfo{author}{Michael \surnamestart Barr\surnameend} (\bibinfo{year}{1992}):
  \emph{\bibinfo{title}{Algebraically compact functors}}.
\newblock {\sl \bibinfo{journal}{Journal of Pure and Applied Algebra}}
  \bibinfo{volume}{82}(\bibinfo{number}{3}), pp. \bibinfo{pages}{211--231},
  \doi{10.1016/0022-4049(92)90169-G}.

\bibitemdeclare{article}{Cockett}
\bibitem{Cockett}
\bibinfo{author}{J.R.B. \surnamestart Cockett\surnameend} \&
  \bibinfo{author}{Stephen \surnamestart Lack\surnameend}
  (\bibinfo{year}{2002}): \emph{\bibinfo{title}{Restriction categories I:
  categories of partial maps}}.
\newblock {\sl \bibinfo{journal}{Theoretical computer science}}
  \bibinfo{volume}{270}(\bibinfo{number}{1-2}), pp. \bibinfo{pages}{223--259},
  \doi{10.1016/S0304-3975(00)00382-0}.

\bibitemdeclare{article}{CNOT}
\bibitem{CNOT}
\bibinfo{author}{J.R.B. \surnamestart Cockett\surnameend}, \bibinfo{author}{Cole
  \surnamestart Comfort\surnameend} \& \bibinfo{author}{Priyaa \surnamestart
  Srinivasan\surnameend} (\bibinfo{year}{2018}): \emph{\bibinfo{title}{The
  Category {CNOT}}}.
\newblock {\sl \bibinfo{journal}{Electronic Proceedings in Theoretical Computer
  Science}} \bibinfo{volume}{266}, pp. \bibinfo{pages}{258--293},
  \doi{10.4204/eptcs.266.18}.

\bibitemdeclare{incollection}{fredkin2002conservative}
\bibitem{fredkin2002conservative}
\bibinfo{author}{Edward \surnamestart Fredkin\surnameend} \&
  \bibinfo{author}{Tommaso \surnamestart Toffoli\surnameend}
  (\bibinfo{year}{2002}): \emph{\bibinfo{title}{Conservative logic}}.
\newblock In: {\sl \bibinfo{booktitle}{Collision-based computing}},
  \bibinfo{publisher}{Springer}, pp. \bibinfo{pages}{47--81},
  \doi{10.1007/BF01857727}.

\bibitemdeclare{phdthesis}{Giles}
\bibitem{Giles}
\bibinfo{author}{Brett \surnamestart Giles\surnameend} (\bibinfo{year}{2014}):
  \emph{\bibinfo{title}{An investigation of some theoretical aspects of
  reversible computing}}.
\newblock Ph.D. thesis, \bibinfo{school}{University of Calgary},
  \doi{10.5072/PRISM/24917}.

\bibitemdeclare{inproceedings}{inverse2}
\bibitem{inverse2}
\bibinfo{author}{Robert \surnamestart Gl{\"u}ck\surnameend} \&
  \bibinfo{author}{Robin \surnamestart Kaarsgaard\surnameend}
  (\bibinfo{year}{2017}): \emph{\bibinfo{title}{A categorical foundation for
  structured reversible flowchart languages}}.
\newblock In: {\sl \bibinfo{booktitle}{Proceedings of the 33rd Conference on
  the Mathematical Foundations of Programming Semantics (MFPS XXXIII)}},
  \doi{10.1016/j.entcs.2018.03.021}.

\bibitemdeclare{article}{Gottesman}
\bibitem{Gottesman}
\bibinfo{author}{Daniel \surnamestart Gottesman\surnameend}
  (\bibinfo{year}{1997}): \emph{\bibinfo{title}{Stabilizer codes and quantum
  error correction}}.
\newblock {\sl \bibinfo{journal}{arXiv preprint \arxiv{quant-ph/9705052}}},
  \doi{10.7907/rzr7-dt72}.

\bibitemdeclare{inproceedings}{quipper}
\bibitem{quipper}
\bibinfo{author}{Alexander~S \surnamestart Green\surnameend},
  \bibinfo{author}{Peter~LeFanu \surnamestart Lumsdaine\surnameend},
  \bibinfo{author}{Neil~J \surnamestart Ross\surnameend},
  \bibinfo{author}{Peter \surnamestart Selinger\surnameend} \&
  \bibinfo{author}{Beno{\^\i}t \surnamestart Valiron\surnameend}
  (\bibinfo{year}{2013}): \emph{\bibinfo{title}{An introduction to quantum
  programming in Quipper}}.
\newblock In: {\sl \bibinfo{booktitle}{International Conference on Reversible
  Computation}}, \bibinfo{organization}{Springer}, pp.
  \bibinfo{pages}{110--124}, \doi{10.1007/978-3-642-38986-3\_10}.

\bibitemdeclare{inproceedings}{Iwama}
\bibitem{Iwama}
\bibinfo{author}{Kazuo \surnamestart Iwama\surnameend}, \bibinfo{author}{Yahiko
  \surnamestart Kambayashi\surnameend} \& \bibinfo{author}{Shigeru
  \surnamestart Yamashita\surnameend} (\bibinfo{year}{2002}):
  \emph{\bibinfo{title}{Transformation rules for designing CNOT-based quantum
  circuits}}.
\newblock In: {\sl \bibinfo{booktitle}{Proceedings of the 39th annual Design
  Automation Conference}}, \bibinfo{organization}{ACM}, pp.
  \bibinfo{pages}{419--424}, \doi{10.1145/513918.514026}.

\bibitemdeclare{article}{Lafont}
\bibitem{Lafont}
\bibinfo{author}{Yves \surnamestart Lafont\surnameend} (\bibinfo{year}{2003}):
  \emph{\bibinfo{title}{Towards an algebraic theory of boolean circuits}}.
\newblock {\sl \bibinfo{journal}{Journal of Pure and Applied Algebra}}
  \bibinfo{volume}{184}(\bibinfo{number}{2-3}), pp. \bibinfo{pages}{257--310},
  \doi{10.1016/S0022-4049(03)00069-0}.

\bibitemdeclare{article}{realize}
\bibitem{realize}
\bibinfo{author}{T~\surnamestart Monz\surnameend},
  \bibinfo{author}{K~\surnamestart Kim\surnameend},
  \bibinfo{author}{W~\surnamestart H{\"a}nsel\surnameend},
  \bibinfo{author}{M~\surnamestart Riebe\surnameend},
  \bibinfo{author}{AS~\surnamestart Villar\surnameend},
  \bibinfo{author}{P~\surnamestart Schindler\surnameend},
  \bibinfo{author}{M~\surnamestart Chwalla\surnameend},
  \bibinfo{author}{M~\surnamestart Hennrich\surnameend} \&
  \bibinfo{author}{R~\surnamestart Blatt\surnameend} (\bibinfo{year}{2009}):
  \emph{\bibinfo{title}{Realization of the quantum Toffoli gate with trapped
  ions}}.
\newblock {\sl \bibinfo{journal}{Physical review letters}}
  \bibinfo{volume}{102}(\bibinfo{number}{4}), p. \bibinfo{pages}{040501},
  \doi{10.1103/PhysRevLett.102.040501}.

\bibitemdeclare{article}{realize2}
\bibitem{realize2}
\bibinfo{author}{Matthew~D \surnamestart Reed\surnameend},
  \bibinfo{author}{Leonardo \surnamestart DiCarlo\surnameend},
  \bibinfo{author}{Simon~E \surnamestart Nigg\surnameend},
  \bibinfo{author}{Luyan \surnamestart Sun\surnameend}, \bibinfo{author}{Luigi
  \surnamestart Frunzio\surnameend}, \bibinfo{author}{Steven~M \surnamestart
  Girvin\surnameend} \& \bibinfo{author}{Robert~J \surnamestart
  Schoelkopf\surnameend} (\bibinfo{year}{2012}):
  \emph{\bibinfo{title}{Realization of three-qubit quantum error correction
  with superconducting circuits}}.
\newblock {\sl \bibinfo{journal}{Nature}}
  \bibinfo{volume}{482}(\bibinfo{number}{7385}), p. \bibinfo{pages}{382},
  \doi{10.1038/nature10786}.

\bibitemdeclare{inproceedings}{Shor}
\bibitem{Shor}
\bibinfo{author}{Peter~W \surnamestart Shor\surnameend} (\bibinfo{year}{1996}):
  \emph{\bibinfo{title}{Fault-tolerant quantum computation}}.
\newblock In: {\sl \bibinfo{booktitle}{Proceedings of 37th Conference on Foundations of Computer Science}}, \bibinfo{organization}{IEEE}, pp.
  \bibinfo{pages}{56--65}, \doi{10.1109/SFCS.1996.548464}.

\bibitemdeclare{article}{ZX}
\bibitem{ZX}
\bibinfo{author}{Renaud \surnamestart Vilmart\surnameend}
  (\bibinfo{year}{2018}): \emph{\bibinfo{title}{A ZX-Calculus with Triangles
  for Toffoli-Hadamard, Clifford+ T, and Beyond}}.
\newblock {\sl \bibinfo{journal}{arXiv preprint arXiv:\arxiv{1804.03084}}}.

\end{thebibliography}

\appendix

\section{Proof of Lemma \ref{lemma:functorial}}
\label{appendix:functorial}

	\begin{enumerate}[label={\bfseries [CNOT.\arabic*]:}, wide = 0pt, leftmargin = 2em]
	\item
	This follows immediately from \ref{TOF.14}.
	
	\item
	{\hspace*{1cm}
		$

		$
	}
	
	\item
	This follows immediately from \ref{TOF.8}.
	
	\item
	{\hspace*{1cm}
		$
		%
		$
	}
	
	\item
	This follows immediately from \ref{TOF.9}.
\end{enumerate}

\section{Proof of Proposition \ref{lemma:zipper}}
\label{appendix:zipper}

First, we show how $\tof$ gates can be ``pushed past'' each other with a ``trailing'' side-effect:

\begin{lemma}
	\label{lemma:pushingGates}
	\label{lemma:GeneralizedPushingGates}
	
	$ $\\
	\begin{multicols}{2}
		\begin{enumerate}[label = {\em (\roman*)}]
			\item
			$
			%
	\end{align*}
	
\end{proof}

Now, we prove Proposition \ref{lemma:zipper}

\begin{proof}
	
	The proof is by a simultaneous induction for claims {(\em i)} and {(\em ii)} on the number of control wires, $n$, to unzip and the number of control wires being pushed past a Toffoli gate.   Claim {(\em iii)} follows as a consequence.
	
	For the induction, suppose that $n,k \geq 2$.  The cases when $n=1$ or $n=2$ follow as a consequence.
	
	\begin{itemize}
		\item		
		The base cases of claim {\em (i)} follows by the definition of the $\cnot_n$ gate.  	The base cases of claim {\em (ii)} Lemma is precisely Lemma \ref{lemma:pushingGates}{\em (ii)}.
		
		\item		
		For $n\geq 2$: assume that for any $m>n$ and all $k \leq n$, we can unzip a $\cnot_m$ gate $k$ levels down that and we can push a $\cnot_n$ gate past a $\tof$ gate as follows:
		
		$$

		$$
		
		Consider the two identities for $n+1$; first for the zipper:
		
		\begin{align*}
		%
		\end{align*}

		For the second inductive step, observe:
		
		\begin{align*}
		%
		\end{align*}
		
	\end{itemize}
	
	Observe that the claim holds when not both $k\geq 2$ and $n\geq 2$,  as a consequence of the induction.
	
	For claim {(\em iii)}, we use claims {(\em i)} and {(\em ii)} to observe:
	
	\begin{align*}
	%
	\end{align*}
	
\end{proof}

\section{Proof of Lemma \ref{remark:Iwama}}
\label{appendix:Iwama}

\begin{enumerate}[label={\em (\roman*)}]
	\item
	This is Lemma \ref{lemma:PolySelfInverse}.
	
	\item
	Using Corollary \ref{cor:transpose} it suffices to consider the case where $x$ is the bottom control wire.  Using Proposition \ref{lemma:zipper} {\em (i)} and Lemma \ref{lemma:PolySelfInverse}, observe:
	
	\begin{align*}
	%
	\end{align*}
	
	\item
	The proof follows easily from Axioms \ref{CNOT.3}, \ref{CNOT.4}, \ref{CNOT.5} and \ref{CNOT.6}.
	
	\item
	The proof follows easily from Axioms \ref{CNOT.4} and \ref{CNOT.5}.

	\item
	The proof is by induction on the number of control wires of the second gate.
	
	The base cases are provided by Proposition \ref{lemma:zipper} {\em (ii)} and {\em (iii)}.
	
	Suppose now that the claim holds for all cases in which the second gate has no more than $n$ control wires.  Consider when the second gate has $n+1$ control wires.  
	Using Corollary \ref{cor:transpose} it suffices to consider the case where $y$ is the bottom wire and $x$ is the second bottom wire:
	
	\begin{align*}
	\oa_x^X \oa_y^{Y\sqcup \{x\}} 
	& = \oa_x^X\rhd_z\oa_{z}^ Y\oa_y^{\{z,x\}} \oa_{z}^Y \lhd_z & \hspace*{-48pt} \text{Use Prop. \ref{lemma:zipper} {\em (i)} to unzip }\oa_y^{Y\sqcup \{x\}} \text{ to the bottom}\\
	&=  \rhd_z\oa_{z}^Y \oa_x^X \oa_y^{\{z,x\}}\oa_{z}^Y \lhd_z\\
	&=  \rhd_z\oa_z^Y \oa_y^{X\sqcup\{z\}} \oa_y^{\{z,x\}} \oa_x^X \oa_z^Y \lhd_z & \text{Prop. \ref{lemma:zipper} {\em (ii)}, {\em (iii)}}\\
	&=  \rhd_z\oa_z^ Y \oa_y^{X\sqcup\{z\}} \oa_y^{\{z,x\}} \oa_z^Y \oa_x^X \lhd_z\\
	&=  \rhd_z \oa_z^Y \oa_y^{X\sqcup\{z\}} \oa_z^Y \oa_y^{\{z,x\}}  \oa_y^{Y\sqcup\{x\}} \oa_x^X \lhd_z & \text{Prop. \ref{lemma:zipper} {\em (ii)}, {\em (iii)}}\\
	&=  \rhd_z\oa_z^Y \oa_z^Y \oa_y^{X\sqcup\{z\}} \oa_y^{X\cup Y} \oa_y^{\{z,x\}} \oa_y^{Y\sqcup\{x\}} \oa_x^X \lhd_z  & \text{Ind. Hyp.} \\
	&=  \rhd_z \oa_y^{X\sqcup\{z\}} \oa_y^{X\cup Y} \oa_y^{\{z,x\}} \oa_y^{Y\sqcup\{x\}} \oa_x^X \lhd_z\\
	&=  \rhd_z \oa_y^{X\sqcup\{z\}} \oa_y^{X\cup Y} \oa_y^{\{z,x\}} \lhd_z \oa_y^{Y\sqcup\{x\}} \oa_x^X \\
	&=  \rhd_z \oa_y^{X\cup Y} \lhd_z\oa_y^{Y\sqcup\{x\}} \oa_x^X & \text{Lem. \ref{remark:Iwama} {\em (ii)}} \\
	&=  \rhd_z\lhd_z \oa_y^{X\cup Y} \oa_y^{Y\sqcup\{x\}} \oa_x^X \\
	&=\oa_y^{X\cup Y} \oa_y^{Y\sqcup\{x\}} \oa_x^X \\
	\end{align*}
	
	\item Using Corollary \ref{cor:transpose}, it suffices to observe:
	
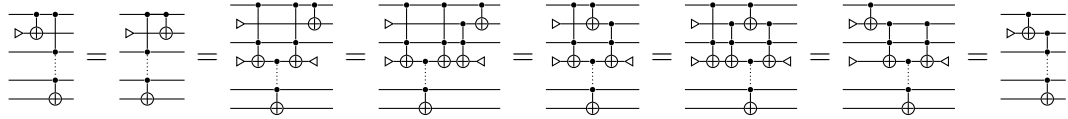
\begin{align*}
	\begin{tikzpicture}
	\begin{pgfonlayer}{nodelayer}
	\node [style=nothing] (0) at (0, -0) {};
	\node [style=nothing] (1) at (0, -1) {};
	\node [style=zeroin] (2) at (0.2500002, -0.5) {};
	\node [style=nothing] (3) at (1.75, -0) {};
	\node [style=nothing] (4) at (1.75, -0.5) {};
	\node [style=nothing] (5) at (1.75, -1) {};
	\node [style=dot] (6) at (0.7499998, -0) {};
	\node [style=oplus] (7) at (0.7499998, -0.5) {};
	\node [style=dot] (8) at (1.25, -0) {};
	\node [style=dot] (9) at (1.25, -1) {};
	\node [style=nothing] (10) at (1.75, -2.25) {};
	\node [style=nothing] (11) at (0, -2.25) {};
	\node [style=oplus] (12) at (1.25, -2.25) {};
	\node [style=dot] (13) at (1.25, -1.75) {};
	\node [style=nothing] (14) at (1.75, -1.75) {};
	\node [style=nothing] (15) at (0, -1.75) {};
	\end{pgfonlayer}
	\begin{pgfonlayer}{edgelayer}
	\draw (0) to (6);
	\draw (6) to (8);
	\draw (8) to (3);
	\draw (4) to (7);
	\draw (7) to (2);
	\draw (7) to (6);
	\draw (8) to (9);
	\draw (5) to (9);
	\draw (9) to (1);
	\draw (11) to (12);
	\draw (12) to (10);
	\draw (14) to (13);
	\draw (13) to (15);
	\draw (12) to (13);
	\draw[style=densely dotted] (9) to (13);
	\end{pgfonlayer}
	\end{tikzpicture}
	=
	\begin{tikzpicture}
	\begin{pgfonlayer}{nodelayer}
	\node [style=nothing] (0) at (0, -0) {};
	\node [style=nothing] (1) at (0, -1) {};
	\node [style=zeroin] (2) at (0.2500002, -0.5) {};
	\node [style=nothing] (3) at (1.75, -0) {};
	\node [style=nothing] (4) at (1.75, -0.5) {};
	\node [style=nothing] (5) at (1.75, -1) {};
	\node [style=dot] (6) at (1.25, -0) {};
	\node [style=oplus] (7) at (1.25, -0.5) {};
	\node [style=dot] (8) at (0.75, -0) {};
	\node [style=dot] (9) at (0.75, -1) {};
	\node [style=nothing] (10) at (1.75, -2.25) {};
	\node [style=nothing] (11) at (0, -2.25) {};
	\node [style=oplus] (12) at (0.75, -2.25) {};
	\node [style=dot] (13) at (0.75, -1.75) {};
	\node [style=nothing] (14) at (1.75, -1.75) {};
	\node [style=nothing] (15) at (0, -1.75) {};
	\end{pgfonlayer}
	\begin{pgfonlayer}{edgelayer}
	\draw (0) to (6);
	\draw (6) to (8);
	\draw (8) to (3);
	\draw (4) to (7);
	\draw (7) to (2);
	\draw (7) to (6);
	\draw (8) to (9);
	\draw (5) to (9);
	\draw (9) to (1);
	\draw (11) to (12);
	\draw (12) to (10);
	\draw (14) to (13);
	\draw (13) to (15);
	\draw (12) to (13);
	\draw [style={densely dotted}] (9) to (13);
	\end{pgfonlayer}
	\end{tikzpicture}
	=
	\begin{tikzpicture}
	\begin{pgfonlayer}{nodelayer}
	\node [style=zeroin] (0) at (0.25, -0) {};
	\node [style=nothing] (1) at (2.75, 0.5) {};
	\node [style=nothing] (2) at (2.75, -0) {};
	\node [style=dot] (3) at (2.25, 0.5) {};
	\node [style=oplus] (4) at (2.25, -0) {};
	\node [style=dot] (5) at (1.25, -1) {};
	\node [style=nothing] (6) at (2.75, -2.25) {};
	\node [style=nothing] (7) at (0, -2.25) {};
	\node [style=oplus] (8) at (1.25, -2.25) {};
	\node [style=dot] (9) at (1.25, -1.75) {};
	\node [style=nothing] (10) at (2.75, -1.75) {};
	\node [style=nothing] (11) at (0, -1.75) {};
	\node [style=nothing] (12) at (0, -0.5) {};
	\node [style=nothing] (13) at (2.75, -0.5) {};
	\node [style=oplus] (14) at (0.75, -1) {};
	\node [style=oplus] (15) at (1.75, -1) {};
	\node [style=dot] (16) at (0.75, -0.5) {};
	\node [style=dot] (17) at (1.75, -0.5) {};
	\node [style=dot] (18) at (0.75, 0.5) {};
	\node [style=nothing] (19) at (0, 0.5) {};
	\node [style=dot] (20) at (1.75, 0.5) {};
	\node [style=zeroin] (21) at (0.25, -1) {};
	\node [style=zeroout] (22) at (2.25, -1) {};
	\end{pgfonlayer}
	\begin{pgfonlayer}{edgelayer}
	\draw (2) to (4);
	\draw (4) to (0);
	\draw (4) to (3);
	\draw (7) to (8);
	\draw (8) to (6);
	\draw (10) to (9);
	\draw (9) to (11);
	\draw (8) to (9);
	\draw [style={densely dotted}] (5) to (9);
	\draw (3) to (1);
	\draw (19) to (18);
	\draw (18) to (20);
	\draw (20) to (3);
	\draw (18) to (16);
	\draw (16) to (14);
	\draw (5) to (14);
	\draw (14) to (21);
	\draw (5) to (15);
	\draw (15) to (22);
	\draw (13) to (17);
	\draw (17) to (16);
	\draw (16) to (12);
	\draw (15) to (17);
	\draw (17) to (20);
	\end{pgfonlayer}
	\end{tikzpicture}
	=
	\begin{tikzpicture}
	\begin{pgfonlayer}{nodelayer}
	\node [style=zeroin] (0) at (0.25, -0) {};
	\node [style=dot] (1) at (1.25, -1) {};
	\node [style=nothing] (2) at (0, -2.25) {};
	\node [style=oplus] (3) at (1.25, -2.25) {};
	\node [style=dot] (4) at (1.25, -1.75) {};
	\node [style=nothing] (5) at (0, -1.75) {};
	\node [style=nothing] (6) at (0, -0.5) {};
	\node [style=oplus] (7) at (0.75, -1) {};
	\node [style=oplus] (8) at (1.75, -1) {};
	\node [style=dot] (9) at (0.75, -0.5) {};
	\node [style=dot] (10) at (1.75, -0.5) {};
	\node [style=dot] (11) at (0.75, 0.5) {};
	\node [style=nothing] (12) at (0, 0.5) {};
	\node [style=dot] (13) at (1.75, 0.5) {};
	\node [style=zeroin] (14) at (0.25, -1) {};
	\node [style=nothing] (15) at (3.25, -0) {};
	\node [style=nothing] (16) at (3.25, -0.5) {};
	\node [style=nothing] (17) at (3.25, 0.5) {};
	\node [style=zeroout] (18) at (2.75, -1) {};
	\node [style=nothing] (19) at (3.25, -2.25) {};
	\node [style=dot] (20) at (2.75, 0.5) {};
	\node [style=oplus] (21) at (2.75, -0) {};
	\node [style=nothing] (22) at (3.25, -1.75) {};
	\node [style=dot] (23) at (2.25, -0) {};
	\node [style=dot] (24) at (2.25, -0.5) {};
	\node [style=oplus] (25) at (2.25, -1) {};
	\end{pgfonlayer}
	\begin{pgfonlayer}{edgelayer}
	\draw (2) to (3);
	\draw (4) to (5);
	\draw (3) to (4);
	\draw [style={densely dotted}] (1) to (4);
	\draw (12) to (11);
	\draw (11) to (13);
	\draw (11) to (9);
	\draw (9) to (7);
	\draw (1) to (7);
	\draw (7) to (14);
	\draw (1) to (8);
	\draw (10) to (9);
	\draw (9) to (6);
	\draw (8) to (10);
	\draw (10) to (13);
	\draw (15) to (21);
	\draw (21) to (20);
	\draw (20) to (17);
	\draw (0) to (23);
	\draw (23) to (21);
	\draw (13) to (20);
	\draw (16) to (24);
	\draw (24) to (10);
	\draw (8) to (25);
	\draw (25) to (18);
	\draw (25) to (24);
	\draw (24) to (23);
	\draw (4) to (22);
	\draw (19) to (3);
	\end{pgfonlayer}
	\end{tikzpicture}
	=
	\begin{tikzpicture}
	\begin{pgfonlayer}{nodelayer}
	\node [style=zeroin] (0) at (0.25, -0) {};
	\node [style=dot] (1) at (1.25, -1) {};
	\node [style=nothing] (2) at (0, -2.25) {};
	\node [style=oplus] (3) at (1.25, -2.25) {};
	\node [style=dot] (4) at (1.25, -1.75) {};
	\node [style=nothing] (5) at (0, -1.75) {};
	\node [style=nothing] (6) at (0, -0.5) {};
	\node [style=oplus] (7) at (0.75, -1) {};
	\node [style=dot] (8) at (0.75, -0.5) {};
	\node [style=dot] (9) at (0.75, 0.5) {};
	\node [style=nothing] (10) at (0, 0.5) {};
	\node [style=zeroin] (11) at (0.25, -1) {};
	\node [style=nothing] (12) at (2.5, -0) {};
	\node [style=nothing] (13) at (2.5, -0.5) {};
	\node [style=nothing] (14) at (2.5, 0.5) {};
	\node [style=zeroout] (15) at (2.25, -1) {};
	\node [style=nothing] (16) at (2.5, -2.25) {};
	\node [style=nothing] (17) at (2.5, -1.75) {};
	\node [style=dot] (18) at (1.25, 0.5) {};
	\node [style=dot] (19) at (1.75, -0.5) {};
	\node [style=oplus] (20) at (1.75, -1) {};
	\node [style=oplus] (21) at (1.25, -0) {};
	\node [style=dot] (22) at (1.75, -0) {};
	\end{pgfonlayer}
	\begin{pgfonlayer}{edgelayer}
	\draw (2) to (3);
	\draw (4) to (5);
	\draw (3) to (4);
	\draw [style={densely dotted}] (1) to (4);
	\draw (10) to (9);
	\draw (9) to (8);
	\draw (8) to (7);
	\draw (1) to (7);
	\draw (7) to (11);
	\draw (8) to (6);
	\draw (4) to (17);
	\draw (16) to (3);
	\draw (21) to (18);
	\draw (22) to (21);
	\draw (20) to (19);
	\draw (19) to (22);
	\draw (1) to (20);
	\draw (20) to (15);
	\draw (13) to (19);
	\draw (19) to (8);
	\draw (0) to (21);
	\draw (22) to (12);
	\draw (14) to (18);
	\draw (18) to (9);
	\end{pgfonlayer}
	\end{tikzpicture}
	=
	\begin{tikzpicture}
	\begin{pgfonlayer}{nodelayer}
	\node [style=zeroin] (0) at (0.25, -0) {};
	\node [style=nothing] (1) at (0, -0.5) {};
	\node [style=oplus] (2) at (0.75, -1) {};
	\node [style=dot] (3) at (0.75, -0.5) {};
	\node [style=dot] (4) at (0.75, 0.5) {};
	\node [style=nothing] (5) at (0, 0.5) {};
	\node [style=zeroin] (6) at (0.25, -1) {};
	\node [style=dot] (7) at (2.25, -0) {};
	\node [style=nothing] (8) at (0, -1.75) {};
	\node [style=nothing] (9) at (3, -1.75) {};
	\node [style=nothing] (10) at (3, -2.25) {};
	\node [style=nothing] (11) at (3, -0) {};
	\node [style=nothing] (12) at (0, -2.25) {};
	\node [style=zeroout] (13) at (2.75, -1) {};
	\node [style=dot] (14) at (1.75, -1.75) {};
	\node [style=oplus] (15) at (1.75, -2.25) {};
	\node [style=dot] (16) at (1.75, 0.5) {};
	\node [style=oplus] (17) at (1.75, -0) {};
	\node [style=dot] (18) at (2.25, -0.5) {};
	\node [style=oplus] (19) at (2.25, -1) {};
	\node [style=dot] (20) at (1.75, -1) {};
	\node [style=nothing] (21) at (3, -0.5) {};
	\node [style=nothing] (22) at (3, 0.5) {};
	\node [style=dot] (23) at (1.25, -0) {};
	\node [style=dot] (24) at (1.25, -0.5) {};
	\node [style=oplus] (25) at (1.25, -1) {};
	\end{pgfonlayer}
	\begin{pgfonlayer}{edgelayer}
	\draw (5) to (4);
	\draw (4) to (3);
	\draw (3) to (2);
	\draw (2) to (6);
	\draw (3) to (1);
	\draw (12) to (15);
	\draw (14) to (8);
	\draw (15) to (14);
	\draw [style={densely dotted}] (20) to (14);
	\draw (14) to (9);
	\draw (10) to (15);
	\draw (17) to (16);
	\draw (7) to (17);
	\draw (19) to (18);
	\draw (18) to (7);
	\draw (20) to (19);
	\draw (19) to (13);
	\draw (21) to (18);
	\draw (7) to (11);
	\draw (22) to (16);
	\draw (23) to (17);
	\draw (18) to (24);
	\draw (24) to (3);
	\draw (2) to (25);
	\draw (0) to (23);
	\draw (23) to (24);
	\draw (24) to (25);
	\draw (25) to (20);
	\draw (16) to (4);
	\end{pgfonlayer}
	\end{tikzpicture}
	=
	\begin{tikzpicture}
	\begin{pgfonlayer}{nodelayer}
	\node [style=zeroin] (0) at (0.25, -0) {};
	\node [style=nothing] (1) at (0, -0.5) {};
	\node [style=nothing] (2) at (0, 0.5) {};
	\node [style=zeroin] (3) at (0.25, -1) {};
	\node [style=dot] (4) at (2.25, -0) {};
	\node [style=nothing] (5) at (0, -1.75) {};
	\node [style=nothing] (6) at (3, -1.75) {};
	\node [style=nothing] (7) at (3, -2.25) {};
	\node [style=nothing] (8) at (3, -0) {};
	\node [style=nothing] (9) at (0, -2.25) {};
	\node [style=zeroout] (10) at (2.75, -1) {};
	\node [style=dot] (11) at (1.75, -1.75) {};
	\node [style=oplus] (12) at (1.75, -2.25) {};
	\node [style=dot] (13) at (2.25, -0.5) {};
	\node [style=oplus] (14) at (2.25, -1) {};
	\node [style=dot] (15) at (1.75, -1) {};
	\node [style=nothing] (16) at (3, -0.5) {};
	\node [style=nothing] (17) at (3, 0.5) {};
	\node [style=dot] (18) at (1.25, -0) {};
	\node [style=dot] (19) at (1.25, -0.5) {};
	\node [style=oplus] (20) at (1.25, -1) {};
	\node [style=oplus] (21) at (0.75, -0) {};
	\node [style=dot] (22) at (0.75, 0.5) {};
	\end{pgfonlayer}
	\begin{pgfonlayer}{edgelayer}
	\draw (9) to (12);
	\draw (11) to (5);
	\draw (12) to (11);
	\draw [style={densely dotted}] (15) to (11);
	\draw (11) to (6);
	\draw (7) to (12);
	\draw (14) to (13);
	\draw (13) to (4);
	\draw (15) to (14);
	\draw (14) to (10);
	\draw (16) to (13);
	\draw (4) to (8);
	\draw (21) to (22);
	\draw (18) to (21);
	\draw (18) to (19);
	\draw (19) to (20);
	\draw (2) to (22);
	\draw (22) to (17);
	\draw (4) to (18);
	\draw (21) to (0);
	\draw (1) to (19);
	\draw (19) to (13);
	\draw (15) to (20);
	\draw (20) to (3);
	\end{pgfonlayer}
	\end{tikzpicture}
	=
	\begin{tikzpicture}
	\begin{pgfonlayer}{nodelayer}
	\node [style=zeroin] (0) at (0.25, -0) {};
	\node [style=nothing] (1) at (0, -0.5) {};
	\node [style=nothing] (2) at (0, 0.5) {};
	\node [style=dot] (3) at (1.25, -0) {};
	\node [style=nothing] (4) at (0, -1.25) {};
	\node [style=nothing] (5) at (1.75, -1.25) {};
	\node [style=nothing] (6) at (1.75, -1.75) {};
	\node [style=nothing] (7) at (1.75, -0) {};
	\node [style=nothing] (8) at (0, -1.75) {};
	\node [style=dot] (9) at (1.25, -1.25) {};
	\node [style=oplus] (10) at (1.25, -1.75) {};
	\node [style=dot] (11) at (1.25, -0.5) {};
	\node [style=nothing] (12) at (1.75, -0.5) {};
	\node [style=nothing] (13) at (1.75, 0.5) {};
	\node [style=oplus] (14) at (0.75, -0) {};
	\node [style=dot] (15) at (0.75, 0.5) {};
	\end{pgfonlayer}
	\begin{pgfonlayer}{edgelayer}
	\draw (8) to (10);
	\draw (9) to (4);
	\draw (10) to (9);
	\draw [style={densely dotted}] (11) to (9);
	\draw (9) to (5);
	\draw (6) to (10);
	\draw (3) to (7);
	\draw (14) to (15);
	\draw (2) to (15);
	\draw (15) to (13);
	\draw (14) to (0);
	\draw (14) to (3);
	\draw (3) to (11);
	\draw (11) to (12);
	\draw (11) to (1);
	\end{pgfonlayer}
	\end{tikzpicture}
	\end{align*}
	
\end{enumerate}

\section{Proof of Proposition \ref{propositoin:TOFDiscreteInverse}}
\label{appendix:discreteInverse}

Most of the proof is inherited from \cite{CNOT} using Lemma \ref{lemma:functorial}.

We first show that $\Delta$ is a natural transformation:

\begin{lemma}
	\label{lemma:DeltaNatural}	
	$\Delta$ is a natural transformation in $\TOF$.
\end{lemma}

\begin{proof}
	We prove that $\Delta$ is a natural transformation by structural induction.  We have only to prove the inductive case for $\tof$ as the cases for the $1$-ancillary bits are proven in \cite[Lemma B.3]{CNOT}:
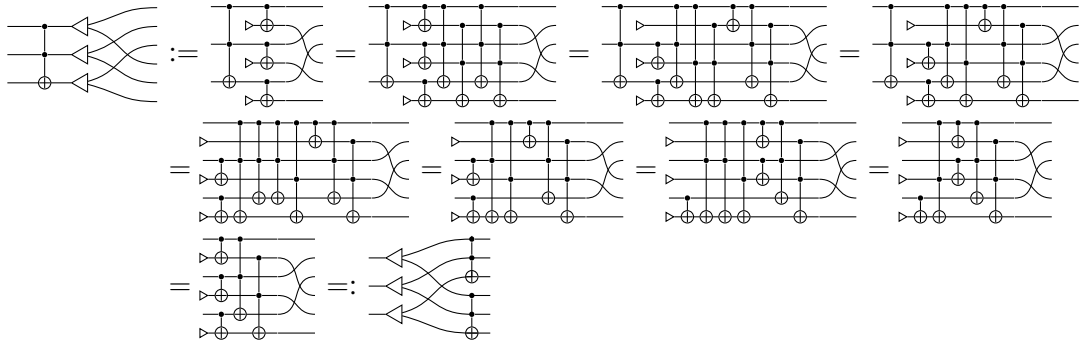
	\begin{align*}

	\end{align*}	
\end{proof}

To prove that $\TOF$ is a discrete inverse category with respect to  $(\_)^\cnv:\TOF^\op\to\TOF$ it must also be shown that \ref{INV.1}, \ref{INV.2} and \ref{INV.3} hold.  As  \ref{INV.1} is immediate it remains to prove \ref{INV.2} and \ref{INV.3}:

\begin{lemma}
	\label{lemma:PartialInverse}	
	For all maps $f$ in $\TOF$ \ref{INV.2} holds, that is $ff^\cnv f=f$.
\end{lemma}

\begin{proof}
	We prove that \ref{INV.2} holds by structural induction.  Again we have only to prove the inductive case for $\tof$ as the cases for the ancillary bits are proven in \cite[Lemma B.14]{CNOT}.
	Suppose inductively that $ff^\cnv f =f$ for some circuit $f$, then we need to show that we can extend $f$ by a $\tof$ gate and preserve the property.  This is almost immediate as:
	
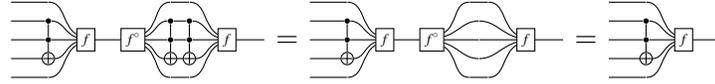
\begin{align*}
	\begin{tikzpicture}
	\begin{pgfonlayer}{nodelayer}
	\node [style=map] (0) at (0.5000001, -0) {$f$};
	\node [style=map] (1) at (1.75, -0) {$f^\cnv$};
	\node [style=map] (2) at (4.25, -0) {$f$};
	\node [style=nothing] (3) at (5.25, -0) {};
	\node [style=dot] (4) at (-0.5000001, -0) {};
	\node [style=dot] (5) at (-0.5000001, 0.5000001) {};
	\node [style=oplus] (6) at (-0.5000001, -0.5000001) {};
	\node [style=nothing] (7) at (-0.5000001, 0.9999999) {};
	\node [style=nothing] (8) at (-1.5, 0.9999999) {};
	\node [style=nothing] (9) at (-1.5, 0.5000001) {};
	\node [style=nothing] (10) at (-1.5, -0) {};
	\node [style=nothing] (11) at (-1.5, -0.5000001) {};
	\node [style=nothing] (12) at (-1.5, -1) {};
	\node [style=nothing] (13) at (-0.5000001, -1) {};
	\node [style=dot] (14) at (2.75, -0) {};
	\node [style=dot] (15) at (2.75, 0.5000001) {};
	\node [style=nothing] (16) at (2.75, -0.9999997) {};
	\node [style=nothing] (17) at (2.75, 1) {};
	\node [style=oplus] (18) at (2.75, -0.5000001) {};
	\node [style=dot] (19) at (3.25, -0) {};
	\node [style=dot] (20) at (3.25, 0.5000001) {};
	\node [style=nothing] (21) at (3.25, -0.9999997) {};
	\node [style=nothing] (22) at (3.25, 1) {};
	\node [style=oplus] (23) at (3.25, -0.5000001) {};
	\end{pgfonlayer}
	\begin{pgfonlayer}{edgelayer}
	\draw (8) to (7);
	\draw [in=150, out=0, looseness=1.00] (7) to (0);
	\draw [in=0, out=-150, looseness=1.00] (0) to (13);
	\draw (13) to (12);
	\draw (6) to (11);
	\draw [in=-165, out=0, looseness=1.00] (6) to (0);
	\draw (0) to (4);
	\draw (4) to (10);
	\draw (9) to (5);
	\draw [in=165, out=0, looseness=1.00] (5) to (0);
	\draw (0) to (1);
	\draw [in=180, out=30, looseness=1.00] (1) to (17);
	\draw (17) to (22);
	\draw [in=150, out=0, looseness=1.00] (22) to (2);
	\draw [in=0, out=-150, looseness=1.00] (2) to (21);
	\draw (21) to (16);
	\draw [in=-30, out=180, looseness=1.00] (16) to (1);
	\draw (1) to (14);
	\draw (14) to (19);
	\draw (19) to (2);
	\draw [in=0, out=165, looseness=1.00] (2) to (20);
	\draw (20) to (15);
	\draw [in=15, out=180, looseness=1.00] (15) to (1);
	\draw [in=180, out=-15, looseness=1.00] (1) to (18);
	\draw (18) to (23);
	\draw (23) to (19);
	\draw (19) to (20);
	\draw (14) to (15);
	\draw (18) to (14);
	\draw [in=-165, out=0, looseness=1.00] (23) to (2);
	\draw (5) to (4);
	\draw (4) to (6);
	\draw (2) to (3);
	\end{pgfonlayer}
	\end{tikzpicture}
	=
	\begin{tikzpicture}
	\begin{pgfonlayer}{nodelayer}
	\node [style=map] (0) at (0.5000001, -0) {$f$};
	\node [style=map] (1) at (1.75, -0) {$f^\cnv$};
	\node [style=map] (2) at (4.25, -0) {$f$};
	\node [style=nothing] (3) at (5.25, -0) {};
	\node [style=dot] (4) at (-0.5000001, -0) {};
	\node [style=dot] (5) at (-0.5000001, 0.5000001) {};
	\node [style=oplus] (6) at (-0.5000001, -0.5000001) {};
	\node [style=nothing] (7) at (-0.5000001, 0.9999999) {};
	\node [style=nothing] (8) at (-1.5, 0.9999999) {};
	\node [style=nothing] (9) at (-1.5, 0.5000001) {};
	\node [style=nothing] (10) at (-1.5, -0) {};
	\node [style=nothing] (11) at (-1.5, -0.5000001) {};
	\node [style=nothing] (12) at (-1.5, -1) {};
	\node [style=nothing] (13) at (-0.5000001, -1) {};
	\node [style=nothing] (14) at (2.75, -0.9999997) {};
	\node [style=nothing] (15) at (2.75, 1) {};
	\node [style=nothing] (16) at (3.25, -0.9999997) {};
	\node [style=nothing] (17) at (3.25, 1) {};
	\node [style=nothing] (18) at (3, -0) {};
	\node [style=nothing] (19) at (3, 0.5000001) {};
	\node [style=nothing] (20) at (3, -0.5000001) {};
	\end{pgfonlayer}
	\begin{pgfonlayer}{edgelayer}
	\draw (8) to (7);
	\draw [in=150, out=0, looseness=1.00] (7) to (0);
	\draw [in=0, out=-150, looseness=1.00] (0) to (13);
	\draw (13) to (12);
	\draw (6) to (11);
	\draw [in=-165, out=0, looseness=1.00] (6) to (0);
	\draw (0) to (4);
	\draw (4) to (10);
	\draw (9) to (5);
	\draw [in=165, out=0, looseness=1.00] (5) to (0);
	\draw (0) to (1);
	\draw [in=180, out=30, looseness=1.00] (1) to (15);
	\draw (15) to (17);
	\draw [in=150, out=0, looseness=1.00] (17) to (2);
	\draw [in=0, out=-150, looseness=1.00] (2) to (16);
	\draw (16) to (14);
	\draw [in=-30, out=180, looseness=1.00] (14) to (1);
	\draw (5) to (4);
	\draw (4) to (6);
	\draw (2) to (3);
	\draw [in=15, out=180, looseness=1.00] (19) to (1);
	\draw [in=180, out=-15, looseness=1.00] (1) to (20);
	\draw [in=-165, out=0, looseness=1.00] (20) to (2);
	\draw [in=0, out=165, looseness=1.00] (2) to (19);
	\draw (2) to (18);
	\draw (18) to (1);
	\end{pgfonlayer}
	\end{tikzpicture}
	=
	\begin{tikzpicture}
	\begin{pgfonlayer}{nodelayer}
	\node [style=map] (0) at (0.5000001, -0) {$f$};
	\node [style=nothing] (1) at (1.5, -0) {};
	\node [style=dot] (2) at (-0.5000001, -0) {};
	\node [style=dot] (3) at (-0.5000001, 0.5000001) {};
	\node [style=oplus] (4) at (-0.5000001, -0.5000001) {};
	\node [style=nothing] (5) at (-0.5000001, 0.9999999) {};
	\node [style=nothing] (6) at (-1.5, 0.9999999) {};
	\node [style=nothing] (7) at (-1.5, 0.5000001) {};
	\node [style=nothing] (8) at (-1.5, -0) {};
	\node [style=nothing] (9) at (-1.5, -0.5000001) {};
	\node [style=nothing] (10) at (-1.5, -1) {};
	\node [style=nothing] (11) at (-0.5000001, -1) {};
	\end{pgfonlayer}
	\begin{pgfonlayer}{edgelayer}
	\draw (6) to (5);
	\draw [in=150, out=0, looseness=1.00] (5) to (0);
	\draw [in=0, out=-150, looseness=1.00] (0) to (11);
	\draw (11) to (10);
	\draw (4) to (9);
	\draw [in=-165, out=0, looseness=1.00] (4) to (0);
	\draw (0) to (2);
	\draw (2) to (8);
	\draw (7) to (3);
	\draw [in=165, out=0, looseness=1.00] (3) to (0);
	\draw (3) to (2);
	\draw (2) to (4);
	\draw (0) to (1);
	\end{pgfonlayer}
	\end{tikzpicture}
	\end{align*}
\end{proof}

To prove that \ref{INV.3} holds, we identify the restriction idempotents of $\TOF$ with what are called latchable circuits.  A map $f$ in a discrete, symmetric monoidal category is called {\bf latchable} when $\Delta(f \ox 1)\Delta^\cnv$  \cite[Definition B.9]{CNOT}.  We already know that latchable circuits commute with each other \cite[Lemma B.10]{CNOT}; therefore, to prove that all circuits of the form $ff^\cnv$ are latchable is to prove that \ref{INV.3} holds.

\begin{lemma}
	\label{lemma:Latchable}
	Circuits of the form $ff^\cnv$ in $\TOF$ are latchable and, thus, \ref{INV.3} holds.
\end{lemma}

\begin{proof}
	We prove that circuits of the form $ff^\cnv$ in $\TOF$ are latchable by structural induction.  We have only to prove the inductive case for $\tof$ as the cases for the $1$-ancillary bits are proven in \cite[Proposition B.12]{CNOT}:
	Suppose that some circuit $f$ is latchable, then we must show inductively that adjoining a Toffoli gate will result in a latchable circuit:
	\begin{align*}
	\begin{tikzpicture}
	\begin{pgfonlayer}{nodelayer}
	\node [style=map] (0) at (0.2500001, -0) {$f$};
	\node [style=map] (1) at (1.25, -0) {$f^\cnv$};
	\node [style=nothing] (2) at (-1.5, -0) {};
	\node [style=nothing] (3) at (-1.5, 0.5000001) {};
	\node [style=nothing] (4) at (-1.5, -0.5000001) {};
	\node [style=nothing] (5) at (-1.5, -1) {};
	\node [style=nothing] (6) at (-1.5, 0.9999999) {};
	\node [style=nothing] (7) at (-0.9999999, 0.9999999) {};
	\node [style=nothing] (8) at (-0.9999999, -1) {};
	\node [style=dot] (9) at (-0.9999999, 0.5000001) {};
	\node [style=dot] (10) at (-0.9999999, -0) {};
	\node [style=oplus] (11) at (-0.9999999, -0.5000001) {};
	\node [style=dot] (12) at (2.5, -0) {};
	\node [style=dot] (13) at (2.5, 0.5000001) {};
	\node [style=nothing] (14) at (2.5, 0.9999999) {};
	\node [style=oplus] (15) at (2.5, -0.5000001) {};
	\node [style=nothing] (16) at (2.5, -1) {};
	\node [style=nothing] (17) at (3, -1) {};
	\node [style=nothing] (18) at (3, -0) {};
	\node [style=nothing] (19) at (3, 0.9999999) {};
	\node [style=nothing] (20) at (3, -0.5000001) {};
	\node [style=nothing] (21) at (3, 0.5000001) {};
	\end{pgfonlayer}
	\begin{pgfonlayer}{edgelayer}
	\draw (9) to (10);
	\draw (10) to (11);
	\draw (11) to (4);
	\draw (2) to (10);
	\draw (3) to (9);
	\draw (6) to (7);
	\draw [in=150, out=0, looseness=1.00] (7) to (0);
	\draw [in=0, out=165, looseness=1.00] (0) to (9);
	\draw (10) to (0);
	\draw [in=-165, out=0, looseness=1.00] (11) to (0);
	\draw [in=-150, out=0, looseness=1.00] (8) to (0);
	\draw (5) to (8);
	\draw (0) to (1);
	\draw (13) to (12);
	\draw (12) to (15);
	\draw [in=180, out=30, looseness=1.00] (1) to (14);
	\draw (14) to (19);
	\draw (21) to (13);
	\draw [in=15, out=180, looseness=1.00] (13) to (1);
	\draw (1) to (12);
	\draw (12) to (18);
	\draw (20) to (15);
	\draw [in=-15, out=180, looseness=1.00] (15) to (1);
	\draw [in=-30, out=180, looseness=1.00] (16) to (1);
	\draw (16) to (17);
	\end{pgfonlayer}
	\end{tikzpicture}
	=
	\begin{tikzpicture}
	\begin{pgfonlayer}{nodelayer}
	\node [style=nothing] (0) at (-2, -0) {};
	\node [style=nothing] (1) at (-2, 0.7499999) {};
	\node [style=nothing] (2) at (-2, -0.7500001) {};
	\node [style=nothing] (3) at (-2, -1.5) {};
	\node [style=nothing] (4) at (-1.5, 1.5) {};
	\node [style=nothing] (5) at (-1.5, -1.5) {};
	\node [style=dot] (6) at (-1.5, 0.7499999) {};
	\node [style=dot] (7) at (-1.5, -0) {};
	\node [style=oplus] (8) at (-1.5, -0.7500001) {};
	\node [style=dot] (9) at (4.5, -0) {};
	\node [style=dot] (10) at (4.5, 0.7499999) {};
	\node [style=nothing] (11) at (4.5, 1.5) {};
	\node [style=oplus] (12) at (4.5, -0.7500001) {};
	\node [style=nothing] (13) at (4.5, -1.5) {};
	\node [style=nothing] (14) at (5, -1.5) {};
	\node [style=nothing] (15) at (5, -0) {};
	\node [style=nothing] (16) at (5, 1.5) {};
	\node [style=nothing] (17) at (5, -0.7500001) {};
	\node [style=nothing] (18) at (5, 0.7499999) {};
	\node [style=fanout] (19) at (-0.9999999, 1.5) {};
	\node [style=fanout] (20) at (-0.9999999, 0.7499999) {};
	\node [style=fanout] (21) at (-0.9999999, -0) {};
	\node [style=fanout] (22) at (-0.9999999, -0.7500001) {};
	\node [style=fanout] (23) at (-0.9999999, -1.5) {};
	\node [style=map] (24) at (0.9999999, 1.25) {$f$};
	\node [style=fanin] (25) at (4, -0) {};
	\node [style=map] (26) at (2, 1.25) {$f^\cnv$};
	\node [style=fanin] (27) at (4, 0.7499999) {};
	\node [style=fanin] (28) at (4, 1.5) {};
	\node [style=fanin] (29) at (4, -0.7500001) {};
	\node [style=fanin] (30) at (4, -1.5) {};
	\node [style=nothing] (31) at (-2, 1.5) {};
	\node [style=nothing] (32) at (1.5, -2.25) {};
	\node [style=nothing] (33) at (1.5, -1.75) {};
	\node [style=nothing] (34) at (1.5, -1.25) {};
	\node [style=nothing] (35) at (1.5, -0.7500001) {};
	\node [style=nothing] (36) at (1.5, -2.75) {};
	\end{pgfonlayer}
	\begin{pgfonlayer}{edgelayer}
	\draw (6) to (7);
	\draw (7) to (8);
	\draw (8) to (2);
	\draw (0) to (7);
	\draw (1) to (6);
	\draw (3) to (5);
	\draw (10) to (9);
	\draw (9) to (12);
	\draw (11) to (16);
	\draw (18) to (10);
	\draw (9) to (15);
	\draw (17) to (12);
	\draw (13) to (14);
	\draw (19) to (4);
	\draw (6) to (20);
	\draw (21) to (7);
	\draw (22) to (8);
	\draw (23) to (5);
	\draw [in=146, out=15, looseness=1.25] (19) to (24);
	\draw [in=15, out=162, looseness=1.25] (24) to (20);
	\draw [in=180, out=0, looseness=1.00] (21) to (24);
	\draw [in=-162, out=18, looseness=1.00] (22) to (24);
	\draw [in=15, out=-146, looseness=1.25] (24) to (23);
	\draw [in=165, out=34, looseness=1.25] (26) to (28);
	\draw [in=18, out=165, looseness=1.25] (27) to (26);
	\draw [in=180, out=0, looseness=1.00] (26) to (25);
	\draw [in=-18, out=162, looseness=1.00] (29) to (26);
	\draw [in=-30, out=165, looseness=1.00] (30) to (26);
	\draw (30) to (13);
	\draw (12) to (29);
	\draw (25) to (9);
	\draw (10) to (27);
	\draw (28) to (11);
	\draw (31) to (4);
	\draw [in=0, out=-165, looseness=1.00] (29) to (32);
	\draw [in=-15, out=180, looseness=1.00] (32) to (22);
	\draw [in=180, out=-15, looseness=1.00] (21) to (33);
	\draw [in=-165, out=0, looseness=1.00] (33) to (25);
	\draw [in=0, out=-165, looseness=1.00] (27) to (34);
	\draw [in=-15, out=180, looseness=1.00] (34) to (20);
	\draw [in=-15, out=180, looseness=1.00] (35) to (19);
	\draw [in=-150, out=0, looseness=1.00] (35) to (28);
	\draw [in=180, out=-15, looseness=1.00] (23) to (36);
	\draw [in=-165, out=0, looseness=1.00] (36) to (30);
	\draw (24) to (26);
	\end{pgfonlayer}
	\end{tikzpicture}
	=
	\begin{tikzpicture}
	\begin{pgfonlayer}{nodelayer}
	\node [style=nothing] (0) at (-2, -0) {};
	\node [style=nothing] (1) at (-2, 0.7500001) {};
	\node [style=nothing] (2) at (-2, -0.7500001) {};
	\node [style=nothing] (3) at (-2, -1.5) {};
	\node [style=fanout] (4) at (-1.5, 1.5) {};
	\node [style=fanout] (5) at (-1.5, 0.7500001) {};
	\node [style=fanout] (6) at (-1.5, -0) {};
	\node [style=fanout] (7) at (-1.5, -0.7500001) {};
	\node [style=fanout] (8) at (-1.5, -1.5) {};
	\node [style=fanin] (9) at (4.5, -0) {};
	\node [style=fanin] (10) at (4.5, 0.7500001) {};
	\node [style=fanin] (11) at (4.5, 1.5) {};
	\node [style=fanin] (12) at (4.5, -0.7500001) {};
	\node [style=fanin] (13) at (4.5, -1.5) {};
	\node [style=nothing] (14) at (-2, 1.5) {};
	\node [style=map] (15) at (2, 2) {$f^\cnv$};
	\node [style=map] (16) at (0.9999999, 2) {$f$};
	\node [style=dot] (17) at (0.9999999, -0.9999999) {};
	\node [style=dot] (18) at (0.9999999, -1.5) {};
	\node [style=dot] (19) at (2, -0.9999999) {};
	\node [style=dot] (20) at (2, -1.5) {};
	\node [style=oplus] (21) at (0.9999999, -2) {};
	\node [style=oplus] (22) at (2, -2) {};
	\node [style=nothing] (23) at (0.9999999, -2.25) {};
	\node [style=nothing] (24) at (2, -2.25) {};
	\node [style=nothing] (25) at (0.9999999, -0.5) {};
	\node [style=nothing] (26) at (2, -0.5) {};
	\node [style=oplus] (27) at (-0.2500001, 1.25) {};
	\node [style=nothing] (28) at (-0.2500001, 2.5) {};
	\node [style=dot] (29) at (-0.2500001, 2.25) {};
	\node [style=dot] (30) at (-0.2500001, 1.75) {};
	\node [style=nothing] (31) at (-0.2500001, 0.7500001) {};
	\node [style=oplus] (32) at (3.25, 1.25) {};
	\node [style=nothing] (33) at (3.25, 2.5) {};
	\node [style=dot] (34) at (3.25, 2.25) {};
	\node [style=dot] (35) at (3.25, 1.75) {};
	\node [style=nothing] (36) at (3.25, 0.7500001) {};
	\node [style=nothing] (37) at (5, -1.5) {};
	\node [style=nothing] (38) at (5, 0) {};
	\node [style=nothing] (39) at (5, 0.7500001) {};
	\node [style=nothing] (40) at (5, 1.5) {};
	\node [style=nothing] (41) at (5, -0.7500001) {};
	\end{pgfonlayer}
	\begin{pgfonlayer}{edgelayer}
	\draw (16) to (15);
	\draw [in=180, out=-15, looseness=1.00] (5) to (17);
	\draw [in=180, out=-15, looseness=1.00] (6) to (18);
	\draw [in=-15, out=180, looseness=1.00] (21) to (7);
	\draw [in=-165, out=0, looseness=1.00] (22) to (12);
	\draw [in=-150, out=0, looseness=1.00] (20) to (9);
	\draw [in=0, out=-150, looseness=1.00] (10) to (19);
	\draw (17) to (19);
	\draw (20) to (18);
	\draw (21) to (22);
	\draw (24) to (23);
	\draw [in=-15, out=180, looseness=1.00] (23) to (8);
	\draw [in=-165, out=0, looseness=1.00] (24) to (13);
	\draw [in=-165, out=0, looseness=1.00] (26) to (11);
	\draw (26) to (25);
	\draw [in=-15, out=180, looseness=1.00] (25) to (4);
	\draw [in=180, out=15, looseness=1.00] (4) to (28);
	\draw [in=30, out=180, looseness=1.00] (29) to (5);
	\draw [in=15, out=180, looseness=1.00] (30) to (6);
	\draw [in=30, out=-165, looseness=0.75] (27) to (7);
	\draw (31) to (8);
	\draw (13) to (36);
	\draw [in=0, out=165, looseness=0.75] (12) to (32);
	\draw [in=0, out=165, looseness=1.00] (9) to (35);
	\draw [in=0, out=165, looseness=1.00] (10) to (34);
	\draw [in=0, out=165, looseness=1.00] (11) to (33);
	\draw [in=22, out=180, looseness=1.00] (33) to (15);
	\draw (15) to (34);
	\draw (35) to (15);
	\draw (32) to (15);
	\draw (36) to (15);
	\draw [in=45, out=-135, looseness=1.00] (16) to (31);
	\draw (27) to (16);
	\draw (16) to (30);
	\draw (29) to (16);
	\draw [in=0, out=158, looseness=1.00] (16) to (28);
	\draw (17) to (18);
	\draw (18) to (21);
	\draw (22) to (20);
	\draw (20) to (19);
	\draw (27) to (30);
	\draw (29) to (30);
	\draw (32) to (35);
	\draw (35) to (34);
	\draw (14) to (4);
	\draw (5) to (1);
	\draw (0) to (6);
	\draw (2) to (7);
	\draw (8) to (3);
	\draw (13) to (37);
	\draw (12) to (41);
	\draw (9) to (38);
	\draw (10) to (39);
	\draw (11) to (40);
	\end{pgfonlayer}
	\end{tikzpicture}
	=
	\begin{tikzpicture}
	\begin{pgfonlayer}{nodelayer}
	\node [style=nothing] (0) at (-2, -0) {};
	\node [style=nothing] (1) at (-2, 0.7500001) {};
	\node [style=nothing] (2) at (-2, -0.7500001) {};
	\node [style=nothing] (3) at (-2, -1.5) {};
	\node [style=fanout] (4) at (-1.5, 1.5) {};
	\node [style=fanout] (5) at (-1.5, 0.7500001) {};
	\node [style=fanout] (6) at (-1.5, -0) {};
	\node [style=fanout] (7) at (-1.5, -0.7500001) {};
	\node [style=fanout] (8) at (-1.5, -1.5) {};
	\node [style=fanin] (9) at (4.5, -0) {};
	\node [style=fanin] (10) at (4.5, 0.7500001) {};
	\node [style=fanin] (11) at (4.5, 1.5) {};
	\node [style=fanin] (12) at (4.5, -0.7500001) {};
	\node [style=fanin] (13) at (4.5, -1.5) {};
	\node [style=nothing] (14) at (-2, 1.5) {};
	\node [style=map] (15) at (2, 2) {$f^\cnv$};
	\node [style=map] (16) at (0.9999999, 2) {$f$};
	\node [style=oplus] (17) at (-0.2500001, 1.25) {};
	\node [style=nothing] (18) at (-0.2500001, 2.5) {};
	\node [style=dot] (19) at (-0.2500001, 2.25) {};
	\node [style=dot] (20) at (-0.2500001, 1.75) {};
	\node [style=nothing] (21) at (-0.2500001, 0.7500001) {};
	\node [style=oplus] (22) at (3.25, 1.25) {};
	\node [style=nothing] (23) at (3.25, 2.5) {};
	\node [style=dot] (24) at (3.25, 2.25) {};
	\node [style=dot] (25) at (3.25, 1.75) {};
	\node [style=nothing] (26) at (3.25, 0.7500001) {};
	\node [style=nothing] (27) at (5, -1.5) {};
	\node [style=nothing] (28) at (5, 0) {};
	\node [style=nothing] (29) at (5, 0.7500001) {};
	\node [style=nothing] (30) at (5, 1.5) {};
	\node [style=nothing] (31) at (5, -0.7500001) {};
	\node [style=nothing] (32) at (1.5, -0) {};
	\node [style=nothing] (33) at (1.5, -0.5000002) {};
	\node [style=nothing] (34) at (1.5, -1) {};
	\node [style=nothing] (35) at (1.5, -2) {};
	\node [style=nothing] (36) at (1.5, -1.5) {};
	\end{pgfonlayer}
	\begin{pgfonlayer}{edgelayer}
	\draw (16) to (15);
	\draw [in=180, out=15, looseness=0.50] (4) to (18);
	\draw [in=30, out=180, looseness=0.50] (19) to (5);
	\draw [in=15, out=180, looseness=0.50] (20) to (6);
	\draw [in=30, out=-165, looseness=0.50] (17) to (7);
	\draw (21) to (8);
	\draw (13) to (26);
	\draw [in=0, out=165, looseness=0.50] (12) to (22);
	\draw [in=0, out=165, looseness=0.50] (9) to (25);
	\draw [in=0, out=165, looseness=0.50] (10) to (24);
	\draw [in=0, out=165, looseness=0.50] (11) to (23);
	\draw [in=22, out=180, looseness=0.50] (23) to (15);
	\draw (15) to (24);
	\draw (25) to (15);
	\draw (22) to (15);
	\draw (26) to (15);
	\draw [in=45, out=-135, looseness=0.50] (16) to (21);
	\draw (17) to (16);
	\draw (16) to (20);
	\draw (19) to (16);
	\draw [in=0, out=158, looseness=0.50] (16) to (18);
	\draw (17) to (20);
	\draw (19) to (20);
	\draw (22) to (25);
	\draw (25) to (24);
	\draw (14) to (4);
	\draw (5) to (1);
	\draw (0) to (6);
	\draw (2) to (7);
	\draw (8) to (3);
	\draw (13) to (27);
	\draw (12) to (31);
	\draw (9) to (28);
	\draw (10) to (29);
	\draw (11) to (30);
	\draw [in=0, out=-165, looseness=1.00] (11) to (32);
	\draw [in=-157, out=0, looseness=1.00] (33) to (10);
	\draw [in=0, out=-165, looseness=1.00] (9) to (34);
	\draw [in=0, out=-171, looseness=1.00] (13) to (35);
	\draw [in=-15, out=180, looseness=1.00] (35) to (8);
	\draw [in=-15, out=180, looseness=1.00] (34) to (6);
	\draw [in=180, out=-15, looseness=1.00] (5) to (33);
	\draw [in=180, out=-15, looseness=1.00] (4) to (32);
	\draw [in=-165, out=0, looseness=1.00] (36) to (12);
	\draw [in=-15, out=180, looseness=1.00] (36) to (7);
	\end{pgfonlayer}
	\end{tikzpicture}
	\end{align*}
\end{proof}

This allows us to Proposition \ref{propositoin:TOFDiscreteInverse}:

\begin{proof}
	Lemmas \ref{lemma:PartialInverse} and \ref{lemma:Latchable} show that $\TOF$ is an inverse category.  Lemma \ref{lemma:DeltaNatural} and the fact that the semi-Frobenius identities hold in $\CNOT$ imply that $(n,\Delta_n,\Delta_n^\cnv)$ forms a natural separable, commutative,  semi-Frobenius algebra for all $n \in \N$.  Thus, $\TOF$ is a discrete inverse category.
\end{proof}

\section{Proof of Lemma \ref{lemma:BehavesAsExpected}}
\label{appendix:BehavesAsExpected}

\begin{multicols}{2}
	\begin{enumerate}[wide = 0pt, leftmargin = 2em]
		\item[\bfseries  $\tof \circ |000\> = |000\>$:]~
		\[

		\end{align*}
		
	\end{enumerate}
\end{multicols}

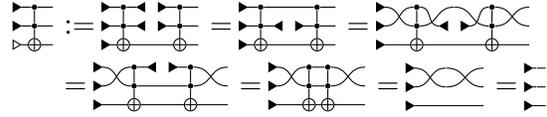
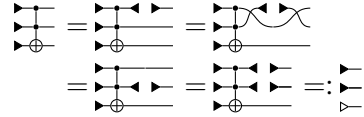
\end{document}